\documentclass[dvipsnames]{article}
\usepackage[round]{natbib}
\let\cite\citep
\usepackage[margin=1in]{geometry}
\usepackage[parfill]{parskip}

\usepackage[utf8]{inputenc} 
\usepackage[T1]{fontenc}    
\usepackage{hyperref}       
\usepackage{url}            
\usepackage{booktabs}      
\usepackage{amsfonts}       
\usepackage{nicefrac}       
\usepackage{microtype}      

\usepackage{color}
\usepackage{amsmath,amssymb,enumerate}
\usepackage{algorithmic}
\usepackage{algorithm}
\usepackage{amsthm}
\usepackage{graphicx}
\usepackage{sidecap}
\usepackage{environ}
\usepackage{placeins}
\usepackage{subfig}
\usepackage[page]{appendix}
\usepackage{adjustbox}

\hypersetup{
    colorlinks=true,
    urlcolor=blue,
    linkcolor=black,
    citecolor=black
}

\usepackage{comment}

\usepackage{caption}
\sidecaptionvpos{figure}{t}

\usepackage{tikz}
\usetikzlibrary{arrows.meta}
\newcommand{\mc}{\mathcal}
\newcommand{\mb}{\mathbb}
\newcommand{\R}{\mb{R}}
\newcommand{\A}{\mathcal{A}}
\newcommand{\X}{\mathcal{X}}
\newcommand{\I}{I}
\renewcommand{\P}{P}
\newcommand{\W}{W}
\newcommand{\onev}{\boldsymbol{1}}
\newcommand{\ones}{\boldsymbol{1}}
\newcommand{\sconst}{\eta}
\newcommand{\supp}{\mathrm{supp}}
\newcommand{\species}{y}
\newcommand{\weights}{w}
\newtheorem{theorem}{Theorem}[section]
\newtheorem{lemma}{Lemma}[section]

\newtheorem{proposition}{Proposition}[section]
\newtheorem{corollary}{Corollary}[section]

\newcommand{\N}{\mathcal{N}}
\title{Evolutionary Game Theory Squared:\\Evolving Agents in Endogenously Evolving Zero-Sum Games}

\makeatletter
\def\blfootnote{\xdef\@thefnmark{}\@footnotetext}
\makeatother

\author{%
  Stratis Skoulakis\thanks{Singapore University of Technology and Design}\rule{.085\textwidth}{0pt}  \\
  \and 
  Tanner Fiez\thanks{University of Washington} \rule{.085\textwidth}{0pt}  \\
  \and 
  Ryann Sim\footnotemark[1] \rule{.02\textwidth}{0pt} \\
  \and  
  Georgios Piliouras\footnotemark[1]~\thanks{Joint last authors} \\
  \and 
  Lillian Ratliff\footnotemark[2]~\footnotemark[3]{\footnote{{Mail: \texttt{efstratios@sutd.edu.sg}, \texttt{fiezt@uw.edu}, \texttt{ryann\_sim@mymail.sutd.edu.sg}, \texttt{georgios@sutd.edu.sg}, \texttt{ratliffl@uw.edu}}}} \\
}
\date{}

\begin{document}

\maketitle

\begin{abstract}
The predominant paradigm in evolutionary game theory and more generally online 
learning in games is based on a clear distinction between a population of \textit{dynamic agents} that interact given a \textit{fixed, static game}.
 In this paper, we move away from the artificial divide between dynamic agents and static games, to introduce and analyze a large class of competitive settings where both the agents and the games they play evolve strategically over time.
We focus on arguably the most archetypal game-theoretic setting---zero-sum games (as well as network generalizations)---and the most studied evolutionary learning dynamic---replicator, the continuous-time analogue of multiplicative weights. Populations of agents compete against each other in a zero-sum competition that itself evolves adversarially to the current population mixture. Remarkably, despite the chaotic coevolution of agents and games,
 we prove that the system exhibits a number of regularities. First, the system has \textit{conservation laws} of an information-theoretic flavor that couple the behavior of all agents and games. Secondly, the system is \textit{Poincar\'{e} recurrent}, with effectively all possible initializations of agents and games lying on recurrent orbits that come arbitrarily close to their initial conditions infinitely often. Thirdly, the \textit{time-average agent behavior and utility converge} to the Nash equilibrium values of the \textit{time-average} game.
 Finally, we provide a polynomial time algorithm to efficiently predict this time-average behavior for any such coevolving network game.  
\end{abstract}

\section{Introduction}
\label{sec:intro}
The problem of analyzing evolutionary learning dynamics in games is of fundamental importance in several fields such as evolutionary game theory~\cite{sandholm2010population}, online learning in games~\cite{cesa2006prediction,nisan2007algorithmic}, and multi-agent systems~\cite{shoham2008multiagent}.
The dominant paradigm in each area is that of evolutionary agents adapting to each others behavior. 
In other words, the dynamism of the
environment of each agent is driven by the other agents, whereas the rules of interaction between the agents, that is, the game, is static. 
\begin{figure}[!htb]
\centering
  \subfloat[][Population $\species$ evolution ]{\includegraphics[width=0.32\linewidth]{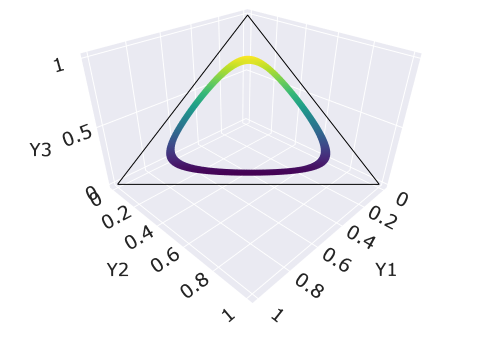}\label{fig:sub1}}\hfill
  \subfloat[][Environment $\weights$ evolution]{\includegraphics[width=0.32\linewidth]{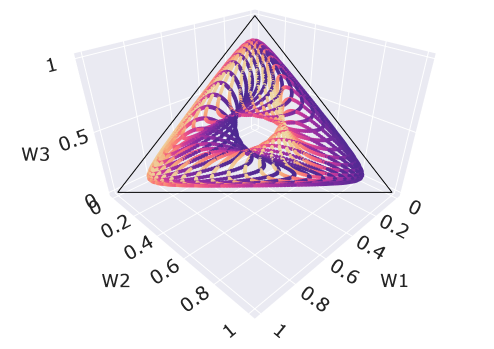}\label{fig:sub2}}\hfill
  \subfloat[][Coevolution of $\species$ and $\weights$]{\includegraphics[width=0.35\linewidth]{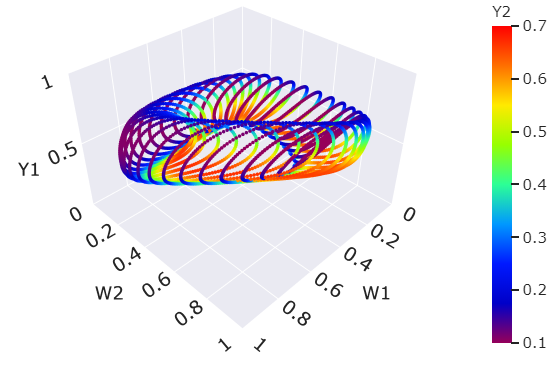}\label{fig:sub3}}\hfill
  \caption{Poincar\'e recurrence in a time-evolving generalized Rock-Paper-Scissors model. 
  }
    \label{fig:main0}
\end{figure}
This separation between \textit{evolving agents} and a \textit{static game} is so standard that it typically goes unnoticed, however, this fundamental restriction does not allow us to capture many applications of interest.  In artificial intelligence~\cite{wang2019evolutionary,garciarena2018evolved,costa2019coegan,miikkulainen2019evolving, wu2019logan,stanley2002evolving} as well as biology,  sociology, and economics~\cite{stewart2014collapse,tilman2020evolutionary,tilman2017maintaining,bowles2003co, weitz2016oscillating}, 
 the rules of interaction can themselves adapt to the collective history of the agent behavior.
 For example, in adversarial learning and curriculum learning~\cite{huang2011adversarial, bengio2009curriculum}, the difficulty of the game can increase over time by exactly focusing on the settings where the agent has performed the weakest.
 Similarly, in biology or economics, if a particular advantageous strategy is used exhaustively by agents, then its relative advantages typically dissipate over time (negative frequency-dependent selection, see~\citealt{heino1998enigma}), which once again drives the need for innovation and exploration. 

In all these cases, the game itself stops being a passive object that the agents act upon, but instead is best thought of as an algorithm itself. Similar to online learning algorithms employed by agents, the game itself may have a memory/state that encodes history. However,
 unlike  online learning algorithms that receive a history or sequence of payoff vectors and output the current behavior (e.g., a probability distribution over actions), an algorithmic game receives as input a history or sequence of agents' behavior and outputs a new payoff matrix. Hence, learning and games are ``dual" algorithmic objects which are coupled in their evolution (Figure~\ref{fig:main0}).

How does one even hope to analyze evolutionary learning in time-evolving games? Once we move away from the safe haven of static games, we lose our prized standard methodology that roughly consists of two steps: i) compute/understand the equilibria of the given game (e.g., Nash, correlated, etc., see~\citealt{nash1951non, aumann1974subjectivity}) and their properties; ii) connect the behavior of learning dynamics  to a target class of equilibria (e.g., convergence).  Indeed, the only prior work to ours, namely by \citet{mai2018cycles}, which considers games larger than $2\times2$, focused on a specific payoff matrix structure based on Rock-Paper-Scissors (RPS) and argued recurrent behavior via a tailored argument that was explicitly designed for the dynamical system in question with no clear connections to game theory. 
We revisit this problem and find a new systematic game-theoretic analysis that generalizes to arbitrary network zero-sum games.  

\subsubsection*{Contributions} We provide a general framework for analyzing learning agents in time-evolving zero-sum games as well as rescaled network generalizations thereof. 
To begin, we develop a novel \textit{reduction} that takes as input time-evolving games and reduces them to a game-theoretic graph that generalizes both graphical zero-sum games and evolutionary zero-sum games. 
 In this generalized but static game, evolving agents and evolving games represent different types of nodes (nodes with and without self-loops) in a graph connected by edge games. 
 The bridge we form between time-evolving games and static network games makes the latter far more interesting than previously thought: \emph{our reduction proves they are sufficiently expressive to capture not only multiple pairwise interactions, but time-varying environments as well.}
 Moreover, by providing a path back to the familiar territory of evolving agents interacting in a static game, the mathematical tools of game theory and dynamical systems theory become available. This allows us to perform a general algorithmic analysis of commonly studied systems from machine learning and biology  previously requiring individualized treatment.

From an algorithmic learning perspective, we focus on the most studied evolutionary learning dynamic: replicator, the continuous-time analogue of the multiplicative weights update. 
Remarkably, despite the chaotic coevolution of agents and games that forces agents to continually innovate, the system can be shown to exhibit a number of regularities. We prove the system is \textit{Poincar\'{e} recurrent}, with effectively all initializations of agents and games lying on recurrent orbits that come arbitrarily close to their initial conditions infinitely often (Figure~\ref{fig:main0}). 
As a crucial component of this result, we demonstrate the dynamics obey information-theoretic \textit{conservation laws} that couple the behavior of all agents and games (Figure \ref{fig:kldivtorus}). Moreover, while the system never equilibrates, the conservation laws allow us to prove the \textit{time-average behavior and utility of the agents converge} to the time-average Nash of their evolving games with bounded regret (Figures~\ref{fig:first_timeaverage} and~\ref{fig:5player} in the Appendix).
Finally, we provide a \textit{polynomial time algorithm} that predicts  these time-average quantities. 

\subsubsection*{Related Work and Technical Novelty} 
Our work relates with the rich previous literature studying the emerging recurrent behavior of replicator dynamics in (network) zero-sum games \cite{piliouras2014persistent,piliouras2014optimization,boone2019darwin, mertikopoulos2018cycles,nagarajanchaos20,perolat2020poincar}. 
All these results are based on the surprising fact that the Kullback-Leibler (KL) divergence between the dynamics produced by the replicator equation and the Nash Equilibrium remains constant. Unfortunately this proof technique is an immediate dead-end for time-evolving zero-sum games (cf. Figure~\ref{fig:klevolve} in Appendix \ref{appsecs:experiments} which shows the KL divergence between the (evolving) strategies and (evolving) Nash equilibrium for the central RPS example from \citealt{mai2018cycles}). In particular, it is not even clear what the static concept of a Nash equilibrium means in this context. Despite this, \citet{mai2018cycles} managed to prove recurrence via constructing an invariant function for this specific example. However, their invariant function relies on the symmetries of the RPS game and has no deeper interpretation or obvious generalization. A key contribution of our work is the development of a novel characterization of a general class of time-evolving games that possess a number of regularities including recurrence, which we demonstrate by deriving an information theoretic invariant.
In particular, this allows us to not only generalize the recurrence results of time-evolving games to a class with much richer and complex interactions than the one studied in \citet{mai2018cycles}, but also provides a naturally interpretable  invariant in such time-evolving games.

\section{Preliminaries and Definitions}
\label{sec:prelims}
In this section, we formalize the concept of polymatrix games, define the replicator dynamics for this class of games, and provide background material on dynamical systems that is relevant to our results. 

\subsubsection*{Polymatrix Games}
An $N$-player \emph{polymatrix game} is defined using an undirected graph $G = (V, E)$ where $V$ corresponds to the set of agents (or players) and $E$ corresponds to the set of edges between agents in which a \emph{bimatrix game} is played between the endpoints~\cite{cai2011minmax}. Each agent $i \in V$ has a set of actions $\A_i = \{1,\ldots,n_i\}$ that can be selected at random from a distribution $x_i$ called a \emph{mixed strategy}. The set of mixed strategies of player $i \in V$ is the standard simplex in $\mb{R}^{n_i}$ and is denoted $\X_i=\Delta^{n_i-1}=\{x_i\in \R^{n_i}_{\geq 0}:\ \sum_{\alpha \in \A_i} x_{i\alpha}=1\}$ where $x_{i\alpha}$ denotes the probability mass on action $\alpha \in \A_i$. The state of the game is then defined by the concatenation of the strategies of all players. We call the set of all possible strategies profiles the \emph{strategy space}, and denote it by $\X = \prod_{i \in V}\X_i$.

The bimatrix game on edge $(i,j)$ is described using a pair of matrices $A^{ij} \in \R^{n_i \times n_j}$ and $A^{ji} \in \R^{n_j \times n_i}$. An entry $A^{ij}_{\alpha\beta}$ for $(\alpha, \beta)\in \A_i\times \A_j$ represents the reward player $i$ obtains for selecting action $\alpha$ given that player $j$ chooses action $\beta$. We note that the graph $G$ may also contain \emph{self-loops}, meaning that an agent $i \in V$ plays a game defined by $A^{ii}$ against itself. The \emph{utility} or \emph{payoff} of agent $i \in V$ under the strategy profile $x \in \X$ is denoted by $u_i(x)$ and corresponds to the sum of payoffs from the bimatrix games the agent participates in. 
The payoff is equivalently expressed as $u_i(x_i,x_{-i})$ when distinguishing between the strategy of player $i$ and all other players $-i$. 
More precisely,
\begin{equation}\label{eq:payoff}
 u_i(x) = \sum\nolimits_{j:(i,j) \in E} x_i^\top A^{ij} x_j.
\end{equation}
We further denote by $u_{i\alpha}(x) = \sum_{j:(i,j) \in E} (A^{ij} x_j)_{\alpha}$ the utility of player $i \in V$ under the strategy profile $x=(\alpha, x_{-i}) \in \X$ for $\alpha \in \A_i$.
The game is called \emph{zero-sum} if $\sum_{i \in V}u_i(x) = 0$  for all $x \in \X$.
Moreover, if there are positive coefficients $\{\sconst_i\}_{i\in V}$ such that 
$\sum_{i \in V}\sconst_i u_i(x)=0$ for all $x \in \X$ and the self-loops are antisymmetric (meaning $A^{ii} = - (A^{ii})^\top$), the game is called \emph{rescaled zero-sum}.

A common notion of equilibrium behavior in game theory is that of a Nash equilibrium, which is defined as a mixed strategy profile $x^\ast\in \X$ such that for each player $i\in V$,
\begin{equation}
    u_i(x_i^\ast, x_{-i}^{\ast})\geq u_{i}(x_i, x_{-i}^\ast), \ \forall x_i\in \X_i.
\end{equation}
We denote the support of $x^\ast_i\in \X_i$ by $\supp(x^\ast_i)=\{\alpha \in \A_i:\ x_{i\alpha}>0\}$. A Nash equilibrium is said to be an \emph{interior} or \emph{fully mixed} Nash equilibrium if $\supp(x_i^\ast)=\A_i\ \forall i\in V$.
\subsubsection*{Replicator Dynamics}
\label{subsec:replicator}
In polymatrix games, \emph{replicator dynamics}~\cite{sandholm2010population} for each $i\in V$ are given by
\begin{equation}\label{eq:replicator}
\dot{x}_{i \alpha} = x_{i \alpha} 
(u_{i\alpha}(x) - u_i(x)),\ \ \forall \alpha\in \A_i.
\end{equation}
We suppress the explicit dependence on time $t$ in the system and do so throughout where clear from context to simplify notation. Moreover, we consider initial conditions on the interior of the simplex.
The replicator dynamics are equivalently given in vector form for each $i\in V$ by the system
\begin{equation}\label{eq:replicator_vector}
 \dot{x}_{i} = x_i \cdot
\big(\sum_{j: (i,j)\in E} A^{ij}x_j - \big(\sum_{j: (i,j)\in E}x_i^\top A^{ij}x_j\big) \cdot \ones \big), 
\end{equation}
where $\ones$ is an $n_i$--dimensional vector of ones and the operator $(\cdot)$ denotes elementwise multiplication. 

For the purpose of analysis, the replicator dynamics in \eqref{eq:replicator} are often translated by a diffeomorphism from the interior of $\X$ to the cumulative payoff space $\mc{C}=\prod_{i\in V}\R^{n_i-1}$, which is defined by a mapping such that
$x_i=(x_{i1}, \ldots, x_{in_i})\mapsto (\ln\tfrac{x_{i2}}{x_{i1}}, \ldots, \ln\tfrac{x_{in_i}}{x_{i1}})$ for each player $i\in V$.

\subsubsection*{Review of Topology of Dynamical Systems}
We now review some concepts from dynamical systems theory that will help us prove Poincar\'e recurrence. 
Further background material can be found in the book of~\citet{alongi2007recurrence}. 

\emph{\textbf{Flows:}}  Consider a differential equation $\dot{x}=f(x)$ on a topological space $X$.
The existence and uniqueness theorem for ordinary differential equations guarantees that there exists a unique continuous function $\phi:\R\times X\to X$, which is termed the \emph{flow}, that satisfies
(i) $\phi(t,\cdot):X\to X$---often denoted $\phi^t:X\to X$---is a homeomorphism for each $t\in \R$,  (ii) $\phi(t+s,x)=\phi(t,\phi(s,x))$ for all $t,s\in\R$ and all $x\in X$, and (iii) for each $x\in X$, $\tfrac{d}{dt}|_{t=0}\phi(t,x)=f(x)$. 
Since the replicator dynamics are Lipschitz continuous, a unique flow $\phi$ of the replicator dynamics  exists.

\emph{\textbf{Conservation of Volume}}: 
The flow $\phi$ of a system of ordinary differential equations is called \emph{volume preserving} if the volume of the image of any set $U \subseteq \R^d$ under $\phi^t$ is preserved. More precisely, for any set $U\subseteq \R^d$,
$\text{vol}(\phi^t(U)) = \text{vol}(U)$.
Whether or not a flow preserves volume can be determined by applying \emph{Liouville's theorem}, which says the flow is volume preserving if and only if the divergence of $f$ at any point $x \in \R^d$ equals zero---that is,
$\text{div}f(x)=\mathrm{tr}(Df(x))=\sum_{i=1}^d \frac{d f(x)}{dx_{i}} = 0.$

 \emph{\textbf{Poincar\'{e} Recurrence:}} 
 If a dynamical system preserves volume and every orbit remains bounded, almost all trajectories return arbitrarily close to their
initial position, and do so infinitely often~\cite{poincare1890probleme}. Given a flow $\phi^t$ on a topological space $X$, a point $x\in X$ is \emph{nonwandering} for $\phi^t$ if for each open neighborhood $U$ containing $x$, there exists $T>1$ such that $U\cap \phi^T(U)\neq \emptyset$. The set of all nonwandering points for $\phi^t$, called the \emph{nonwandering set}, is denoted $\Omega(\phi^t)$. 

\begin{theorem}[Poincar\'{e} Recurrence~\cite{poincare1890probleme}]\label{t:poincare}
If a flow preserves volume and has
only bounded orbits, then for each open set almost all orbits intersecting the set
intersect it infinitely often: if $\phi^t$ is a volume preserving flow on a bounded set $Z\subset \R^d$, then $\Omega(\phi^t)=Z$.
\end{theorem}

\section{Studying Doubly Evolutionary Processes via
Polymatrix Games}
\label{sec:evolution}

Numerous applications from  artificial intelligence (AI) and machine learning (ML) to biology cast competition between populations (e.g., neural networks/algorithms or species/agents) and the environment (e.g., hyperparameters/network configurations or resources) as a time-evolving dynamical system.
The basic abstraction takes the form of a population $y$ of \emph{species} which evolve dynamically in time as a function of itself and some \emph{environment} parameters $w$ whose evolution, in turn, depends on $y$. 
We now review models from each application and then connect a broad class of time-evolving dynamical systems to static polymatrix games. This reduction provides a path toward analyzing complex non-stationary dynamics using tools developed for the typical static game formulation.

\subsubsection*{Doubly Evolutionary Behavior in AI and ML}  
Evolutionary game theory methods for training generative adversarial networks commonly exhibit time-evolving dynamic behavior and there is a pair of predominant doubly evolutionary process models~\cite{costa2020using,wang2019evolutionary,garciarena2018evolved,costa2019coegan,miikkulainen2019evolving}.
In the first formulation,
\citet{wang2019evolutionary} describe training the generator network, with parameters $y$, via a gradient-based algorithm composed of \emph{variation}, \emph{evaluation}, and \emph{selection}. The discriminator network, with parameters $w$ updated via gradient-based learning, is modeled as the environment operating in a feedback loop with $y$.
The second model is such that the generator and discriminator are different species  (or \emph{modules}) in the population $y$ which follows evolutionary dynamics, 
and network hyperparameters (or \emph{chromosomes}) $w$ evolve in time as a function of $y$~\cite{garciarena2018evolved,costa2019coegan,miikkulainen2019evolving}.
We connect further to AI and ML applications in the discussion where we highlight exciting future directions.
 % (Section~\ref{sec:conclusion}).

\subsubsection*{Doubly Evolutionary Behavior in Biology} There are also two common formulations emerging in biology. 
In the first, the focus is on the level of coordination in a population as a function of evolving environmental variables. The prevailing model is comprised of  replicator dynamics $\dot{y}=y(1-y)((A(w)y)_1-(A(w)y)_2)$ in which a population of two species $y$ plays a prisoner's dilemma (PD) game against themselves in a setting where the payoff matrix $A(w)$ depends on an environment variable $w$ which, in turn, depends on the population via $\dot{w}=w(1-w)G(y)$ where $G(y)$  is a feedback mechanism describing when environmental degradation or enhancement occurs as a function of $y$ \cite{weitz2016oscillating,tilman2020evolutionary,tilman2017maintaining,lade2013regime}; e.g., in \citet{weitz2016oscillating}, $G(y)$ takes the form $\theta y-(1-y)$ for some $\theta>0$ which represents  the ratio of the enhancement rate to degradation rate of `cooperators' and `defectors' in the time-evolving PD game.
In the second formulation, the focus is on studying how competition among species is modulated by resource availability.
Indeed, from a biological perspective,~\citet{mai2018cycles} argue that the environment parameters $w$ on which a population 
$y$ of $n$ antagonistic species depend are not constant, but rather evolve over time. Since the species fitness depends on the environment, the game among the species is also time-varying. 
The adopted 
model of the dynamic behavior with initial conditions on the interior of the simplex for both $w$ and $y$
is given for each $i \in \{1,\dots, n\}$ by
\begin{equation} 
\begin{aligned}
\dot{\weights}_i &= \weights_i \sum_{j=1}^n\weights_j(\species_j - \species_i)\\ 
\dot{\species}_i &= \species_i \big((\P(\weights) \species)_i - \species^\top \P(\weights) \species \big)
\end{aligned}
\label{eq:replicator_evolution}
\end{equation}
where 
$\P(w) =\P + 
\mu \W$ for $\mu>0$ with $P$ defined as the generalized RPS payoff matrix 
\[\P=\begin{pmatrix}
0 & -1 & 0 & \cdots & 0 & 0 & 1 \\
1 & 0 & -1 & \cdots & 0 & 0 & 0\\
\vdots & \vdots & \vdots & \vdots & \vdots & \vdots & \vdots\\
0 & 0 & 0 & \cdots & 1 & 0 & -1\\
-1 & 0 & 0 & \cdots & 0 & 1 & 0
\end{pmatrix},\]
and the environmental variations matrix
\[\W=\begin{pmatrix}
0 & w_1 - w_2 &  \cdots &  w_1 - w_n\\
w_2 - w_1 & 0 & \cdots & w_2-w_n\\
\vdots & \vdots  & \vdots & \vdots\\
w_n - w_1 &w_n-w_2 & \cdots & 0\\
\end{pmatrix}.\]

\subsubsection*{Reducing Time-Evolving RPS to a Polymatrix Game}
\label{sec:timeevolveRPS}
\citet{mai2018cycles} studied the dynamical system in~\eqref{eq:replicator_evolution} and showed it exhibits a special type of cyclic behavior: \emph{Poincar\'{e} recurrence}. By capturing the evolution of the environment (dynamics of the payoff matrix) as additional players that dynamically change their strategies,
we reduce the coevolution of $w$ and $y$ to 
a \textit{static polymatrix game} of greater dimensionality (greater number of players).
Given this reduction, Theorem~\ref{t:main}, which establishes the Poincar\'{e} recurrence of replicator dynamics in rescaled zero-sum polymatrix games, immediately captures the results of~\citet{mai2018cycles} (see Corollary~\ref{cor:evolving}).

\begin{proposition}
The time-evolving generalized rock-paper-scissors game from~\eqref{eq:replicator_evolution} is equivalent to replicator dynamics in a two-player rescaled zero-sum polymatrix game.
\label{prop:rps}
\end{proposition}
\begin{proof}[Proof Sketch \textnormal{(see Appendix~\ref{appsecs:evolution_proof} for formal proof)}]
The initial condition $\weights(0)$ is on the interior of the simplex and $\sum_{i=1}^n \dot{w}_i=0$. Consequently, $\sum_{i=1}^n w_i(0)= \sum_{i=1}^n w_i(t)= 1$, and we obtain 
\begin{equation*}
\dot{\weights}_i = \weights_i \sum_{j=1}^n\weights_j(\species_j - \species_i)
= \weights_i\big(-\species_i +  \sum_{j=1}^n\weights_j \species_j\big),
\end{equation*}
which is the replicator equation of a node $w$ in a polymatrix game  with payoff matrix $A^{wy} = -\I$.
Using a similar decomposition, we reformulate the $y$ dynamics:
\begin{align*}
\dot{\species}_i 
&= \species_i  \big ((\P \species)_i - \species^\top \P \species\big) + \species_i \big(
\mu w_i - \mu\sum_{j=1}^n  w_j \species_j\big).
\end{align*}
This corresponds to the replicator equation of node $\species$ playing against itself with $A^{yy} = \P$ and against $w$ with $A^{yw} = \mu \I$. The game is rescaled zero-sum with $\eta_\species = 1$ and $\eta_\weights = \mu$.
\end{proof}
\subsubsection*{Generalized Reduction}
The previous reduction generalizes to a class of time-evolving games defined by a set of populations $y=(y_1,\ldots, y_{n_y})$ and environments $w=(w_1,\ldots, w_{n_w})$, where $y_\ell\in \Delta^{n-1}$ for each $\ell\in\{1,\ldots, n_y\}$ and $w_k\in \Delta^{n-1}$ for each $k\in \{1,\ldots, n_w\}$. Environments coevolve with only populations and not other environments, while any population coevolves only with  environments and itself. Let $\mc{N}_k^w$ be the set of populations which coevolve with $w_k$ and $\mc{N}_\ell^y$ be the set of environments which coevolve with $y_\ell$. The time-evolving dynamics for each environment $k$ and population $\ell$ are given componentwise by
\begin{align}
 \dot{\weights}_{k,i} &=\weights_{k,i} \sum_{\ell\in \N^w_k}\sum_{j}\weights_{k,j}\left(  (A^{k,\ell}\species_{\ell})_{i}-(A^{k,\ell}\species_{\ell})_{j}\right),\label{eq:fbevolutionwa}\\
 \dot{\species}_{\ell,i}&= \species_{\ell,i}  \left(
(\P_\ell({w}) \species_\ell)_{i} - \species_\ell^\top \P_\ell(\weights) \species_\ell 
\right), 
\label{eq:fbevolutionya}
\end{align}
where 
$\P_\ell(w)=\P_\ell+\sum_{k\in \N^y_\ell} \W^{\ell,k}$ with $\P_\ell\in \mathbb{R}^{n\times n}$ and $W^{\ell,k}\in \mathbb{R}^{n\times n}$ is defined such that the $(i,j)$--th entry is $(A^{\ell,k}w_k)_i-(A^{\ell,k}w_k)_j$. 

Despite the complex nature of this dynamical system, we can show that it is equivalent to replicator dynamics in a polymatrix game.
The proof of this result is in Appendix~\ref{appsecs:evolution_proof}.
\begin{theorem}
\label{thm:evolution_general}
Any time-evolving system defined by the dynamics in (\ref{eq:fbevolutionwa}-\ref{eq:fbevolutionya}) is equivalent to replicator dynamics in a polymatrix game. 
\end{theorem}

 The expressive power we gain from this reduction permits us to efficiently describe and characterize coevolutionary processes of higher complexity than past work since we can return to the familiar territory of analyzing dynamic agents in static games. In what follows we focus on providing theoretical results for the subclass of time-evolving systems which reduce to a rescaled zero-sum game. However, this reduction is of independent interest since it can prove useful for future work analyzing the class of general-sum games after the behavior of network zero-sum games and rescaled generalizations are well understood.

\section{Poincar\'{e} Recurrence} 
\label{sec:recurrence}

In this section, we
show that the replicator dynamics are Poincar\'{e} recurrent in $N$-player rescaled zero-sum polymatrix games with interior Nash equilibria. In particular,
 for almost all initial conditions $x(0) \in \X$, the replicator dynamics will return arbitrarily close to $x(0)$ an infinite number of times. 

\begin{theorem}\label{t:main}
The replicator dynamics given in~\eqref{eq:replicator} are Poincar\'{e} recurrent in any $N$-player rescaled zero-sum polymatrix game that has an interior Nash equilibrium.
\end{theorem}

\citet{boone2019darwin}, the closest known result, prove replicator dynamics are Poincar\'{e} recurrent in $N$-player \emph{pairwise zero-sum} polymatrix games with an interior Nash equilibria, which requires $A^{ij} = -(A^{ji})^{\top}$ for every  $(i, j)\in E$. Our extension to $N$-player \emph{rescaled} zero-sum polymatrix games is a far more general characterization of the Poincar\'{e} recurrence of replicator dynamics since there are no explicit restrictions on the edge games and the polymatrix game itself need not even be strictly zero-sum. 
The significance of this result is further enhanced by the connection developed in Section~\ref{sec:evolution} between a  class of time-evolving games and $N$-player rescaled zero-sum polymatrix games. 
As a concrete example, 
given the reduction of Proposition~\ref{prop:rps}, Theorem~\ref{t:main} recovers the work of~\citet{mai2018cycles}.
 \begin{corollary}
 \label{cor:evolving}
The time-evolving generalized rock-paper-scissors game in~\eqref{eq:replicator_evolution} is Poincar\'{e} recurrent.
 \end{corollary}

The following proof sketch provides intuition that highlights our analysis techniques and we defer the finer points to Appendix~\ref{appsecs:recurrence_proof}. It is worth noting that the technical results we prove in order to show the system is Poincar\'{e} recurrent, namely volume preservation and the bounded orbits property, are themselves independently important as they provide conservation laws that couple the behavior of agents. In fact, they are fundamental to showing that while the system never equilibrates, the time-average dynamics and utility converge to the Nash equilibrium and its utility.

\subsubsection*{Overview of Proof Methods}
To prove Poincar\'{e} recurrence, we need to show the flow corresponding to the system of ordinary differential equations in~\eqref{eq:replicator} is volume preserving and has bounded orbits (cf.~Theorem~\ref{t:poincare}). Notice that the flow of \eqref{eq:replicator}  always has bounded orbits since
$x_{i\alpha} \geq 0$ and $\sum_{\alpha \in \A_i}x_{i\alpha}(t)=1 \ \forall \ i \in V$, however proving the volume preserving property is not as straightforward.
To show volume preservation, we transform the dynamics via a \emph{canonical transformation}. 
Indeed, we prove Poincar\'{e} recurrence of the flow of a system of ordinary differential equations that is diffeomorphic  to
the flow of the replicator equation.  
Given $x \in \X$, consider  the transformed variable $z \in \mathbb{R}^{n_1 + \cdots + n_N-N}$ defined by
\begin{equation}\label{eq:transformation}
z_{i}  = \big(\ln\tfrac{x_{i2}}{ x_{i1}}
, \ldots, \ln \tfrac{x_{in_i}}{ x_{i1}}\big), \ \forall i \in V.
\end{equation}
Given the vector $z_i$, the components of $x_i$ are given by
$x_{i\alpha}=  e^{z_{i\alpha}}/(\sum_{\ell = 1}^{n_i}e^{z_{i \ell}})$.
Under this transformation,  $\dot{z}=F(z)$ is given componentwise for each $\alpha\in \A_i$ and all $i\in V$ by 
\begin{equation}
\dot{z}_{i \alpha} =F_{i\alpha}(z)
= \frac{\dot{x}_{i\alpha}}{x_{i\alpha}} - 
\frac{\dot{x}_{i1}}{x_{i1}}
=\sum_{j \in V} \sum_{\beta \in \A_j} (A^{ij}_{\alpha \beta}-
 A^{ij}_{1 \beta})  e^{z_{j \beta}}/\sum_{\ell = 1}^{n_j}e^{z_{j \ell}}.
 \label{eq:zdynamics} 
\end{equation}
Observe that $F_{i1} =0$, meaning $z_{i1} =0$ for all time. To show Poincar\'{e} recurrence of \eqref{eq:replicator}, we prove two key properties: \emph{(i)} the flow of $\dot{z}$
is volume preserving, meaning the trace Jacobian of the respective vector field $\dot{z}=F(z)$ is zero, and,
\emph{(ii)} $\dot{z}$ has bounded orbits from any interior initial condition. Then, the Poincar\'{e} recurrence of $\dot{z}$, and consequently $\dot{x}$,  follows from Theorem~\ref{t:poincare}.

\subsubsection*{Conservation of Volume}
We show that the trace of the vector field $F(z)$ is zero, which then from {Liouville's theorem} 
guarantees $\dot{z}$, as defined in~\eqref{eq:zdynamics}, is volume preserving. 
\begin{lemma}\label{l:divergence}
For any $N$-player rescaled zero-sum polymatrix game, 
\[\mathrm{tr}(DF(z)) = \sum_{i \in V}\sum_{\alpha \in \A_i}\frac{d }{dz_{i\alpha}}F_{i\alpha}(z)=0.\]
\end{lemma}
The proof of Lemma~\ref{l:divergence} crucially relies on the fact the self-loops are antisymmetric, $(A^{ii})^\top=-A^{ii}$. 

\subsubsection*{Bounded Orbits}
\noindent In order to prove that the orbits from any initial interior point $z(0)$ are bounded, we show that for any initial interior point $x(0)$, the orbit produced by the replicator dynamics stays on the interior of the simplex, that is, there exists a fixed parameter $\epsilon >0$ such that for any agent $i \in V$ and strategy $\alpha \in \A_i$, $\epsilon \leq x_{i\alpha} \leq 1- \epsilon$. Then, $|z_{i\alpha}|$ is clearly bounded since $z_{i\alpha} = \ln (x_{i\alpha}/x_{1\alpha})$. 

\begin{lemma}\label{t:invariant}
Consider an $N$-player rescaled zero-sum polymatrix game such that for positive coefficients $\{\sconst_i\}_{i\in V}$, $\sum_{i\in V} \sconst_i  u_i(x) =0 $ for $x \in \X$. If the game admits an interior Nash Equilibrium $x^\ast$, then $\Phi(t) = \sum_{i \in V} \sum_{\alpha \in \A_i}\sconst_i  x_{i\alpha}^\ast \ln x_{i\alpha}$ is time-invariant, meaning $\Phi(t) =\Phi(0)$ for $t\geq 0$. Hence, orbits from any interior initial condition $x(0)$ remain on the interior of the simplex.
\end{lemma}

From the preceding discussion, Lemma~\ref{t:invariant} guarantees orbits from any interior initial condition $z(0)$ remain bounded.
The proof of Lemma~\ref{t:invariant} is the primary novelty in the proof of Theorem~\ref{t:main} and the techniques may be of independent interest. 
To show $\Phi(t)$ is time-invariant, we prove that the time derivative of the function is equal to zero. From the given form of the replicator dynamics and the rescaled zero-sum property of the polymatrix game, we obtain $\dot{\Phi}(t) = \sum_{i\in V}\sum_{j:(i, j)\in E}\eta_i (x_i^\ast)^\top A^{ij}(x_j-x_j^\ast)$ nearly immediately, where the sum over edges describes how the rescaled utility of agent $i \in V$ changes at her equilibrium strategy when the rest of the players are allowed to deviate. To continue, we draw a key connection to a fascinating result regarding the payoff structure of zero-sum polymatrix games.  

\citet{cai2011minmax} proved there exists a payoff preserving transformation from any zero-sum polymatrix game to a pairwise constant-sum polymatrix game. We translate this result to rescaled zero-sum polymatrix games. 
The primary implication is that 
 the change in player $i$'s  rescaled utility at equilibrium when all other players connected to $i$ deviate is equal to the change in player $j$'s rescaled  utility from deviating while all other players connected to $j$ remain in equilibrium. This is a direct consequence of the fact that the game is equivalent to a pairwise constant-sum game.  Explicitly, we prove that $\dot{\Phi}(t) = \sum_{j \in V}\sum_{i:(j, i)\in E}\eta_j(x_j^{\ast}-x_j)^\top A^{ji}x_i^\ast$ and conclude $\dot{\Phi}(t)=0$ since $x^\ast$ is an interior Nash equilibrium, which means $u_{j\alpha}(x^\ast)=u_j(x^\ast)$ for $\alpha\in \A_j$ and any linear combination.

 \begin{proof}[Proof of Theorem~\ref{t:main}]
 The proof follows directly from Lemma~\ref{l:divergence}, Lemma~\ref{t:invariant}, and Theorem~\ref{t:poincare}. Indeed, the dynamics in~\eqref{eq:zdynamics} are Poincar\'{e} recurrent since from Lemma~\ref{l:divergence} they are volume preserving and from Lemma~\ref{t:invariant} the orbits are bounded. This property in the cumulative payoff space carries over to the dynamics in the strategy space from~\eqref{eq:replicator} since the transformation is a  diffeomorphism.
 \end{proof}

\section{Time-Average Behavior, Equilibrium Computation, \& Bounded Regret}
\label{sec:timeaverage}
In this section, we transition away from analyzing the dynamic behavior of replicator dynamics and focus on characterizing the long-term behavior along with its connections to notions of equilibrium and regret. 
We prove that the enduring system behavior is guaranteed to satisfy a number of desirable game-theoretic metrics of consistency and optimality. Moreover, we design a polynomial time algorithm able to predict this behavior. 
The proofs of results from this section are in Appendix~\ref{appsecs:regret}.

While the replicator dynamics exhibit complex dynamics and never equilibriate in rescaled zero-sum polymatrix games with interior Nash equilibrium, the time-average behavior of the dynamics is closely tied to the equilibrium. The following result shows that given the existence of a unique interior Nash equilibrium, the time-average of the replicator dynamics converges to the equilibrium and the time-average utility converges to the utility at the equilibrium.
\begin{theorem}\label{t:avg1}
Consider an $N$-player rescaled zero-sum polymatrix game that admits a unique interior Nash equilibrium $x^\ast$.
The trajectory $x(t)$ produced by replicator dynamics given in~\eqref{eq:replicator} is such that \textbf{i)} $\lim_{t \rightarrow \infty}\frac{1}{t}\int_{0}^t x(\tau)d\tau = x^\ast$ and \textbf{ii)}
$\lim_{t \rightarrow \infty}\frac{1}{t}\int_{0}^t u_i(x(\tau))d\tau = u_i(x^\ast)$.
\end{theorem}

The preceding result provides a broad generalization of past results that show the time-average of replicator dynamics converges to the unique interior Nash equilibrium in zero-sum bimatrix games~\cite{hofbauer2009time}. We remark that our proof crucially relies on Lemma~\ref{t:invariant} since the trajectory of the dynamics must remain on the interior of the simplex to guarantee there exists a bounded sequence which admits a subsequence that converges to a limit corresponding to the time-average. 

We now provide a polynomial time algorithm that  efficiently predicts the time-average quantities even for an arbitrary networks of players. Linear programming formulations for computing and characterizing the set of Nash equilibria for zero-sum polymatrix games are known~\cite{CCDP16}. The following result extends this formulation to rescaled zero-sum polymatrix games. 
\begin{theorem}
Consider an $N$-player rescaled zero-sum polymatrix game such that for positive coefficients $\{\sconst_i\}_{i\in V}$, $\sum_{i=1}^N \sconst_i  u_i(x) =0 $ for $x \in \X$. The optimal solution of the following linear program is a Nash equilibrium of the game:
\begin{equation*}
 \min_{x \in \X}\  \{\sum_{i=1}^{n} \sconst_i v_i |\ \displaystyle v_i \geq u_{i\alpha}(x),\ \forall \ i \in V, \  \forall\  \alpha \in \A_i\}
\end{equation*}

\label{thm:lp}
\end{theorem}

It cannot be universally expected that an interior equilibrium exists or that players are fully rational and obey a common learning rule. Similarly, players may not always be able to determine an equilibrium strategy \emph{a priori} depending on the information available.
This motivates an evaluation of the trajectory of a player who is oblivious to opponent behavior.
We consider a notion of \emph{regret} for a player. That is, the time-averaged utility difference between the mixed strategies selected along the learning path $t\geq 0$ and the fixed strategy that maximizes the utility in hindsight. Even in polymatrix games (with self-loops), the regret of replicator dynamics stays bounded.
\begin{proposition} Any player following the replicator dynamics~\eqref{eq:replicator} in an $N$-player polymatrix game (with self-loops) achieves an $\mathcal{O}(1/t)$  regret bound independent of the rest of the players. Formally, for every trajectory $x_{-i}(t)$, the regret of player  $i \in V$ is bounded as follows for a player-dependent positive constant $\Omega_i$,
\[\mathrm{Reg}_i(t):=\max_{y\in \X_i}\frac{1}{t}\int_0^t\left[u_i(y, x_{-i}(s))-u_i(x(s))\right]ds\leq \frac{\Omega_i}{t}.\]
\label{thm:regretbdd}
\end{proposition} 
The proof of this proposition mirrors closely more general arguments in~\citet{mertikopoulos2018cycles}. A standalone derivation is provided in the appendix sake of completeness.

\section{Simulations}
\label{sec:simulations}

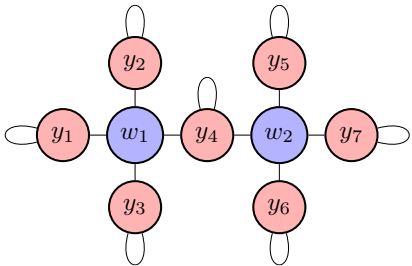
\begin{figure}[t]
    \centering
    \resizebox{0.35\linewidth}{!}{
    \begin{tikzpicture}
    \tikzset{every loop/.style={}}
    \begin{scope}[every node/.style={circle,thick,draw}]
        \node[fill =red!30] (x1) at (-3,2) {$\species_1$};
        \node[fill =red!30] (x2) at (-2,3) {$\species_2$};
        \node[fill =red!30] (x3) at (-2,1) {$\species_3$};
        \node[fill =red!30] (x4) at (-1,2) {$y_4$};
        \node[fill =red!30] (x5) at (0,3) {$y_5$};
        \node[fill =red!30] (x6) at (0,1) {$y_6$};
        \node[fill =red!30] (x7) at (1,2) {$y_7$};
        
        \node[fill =blue!30] (w1) at (-2,2) {$\weights_1$};
        \node[fill =blue!30] (w2) at (0,2) {$\weights_2$};
    \end{scope}
    \begin{scope}[>={Stealth[black]},
                 every node/.style={fill=white,circle},
                 every edge/.style={draw=black}]
          \path [-] (x1) edge (w1);
          \path [-] (x2) edge (w1);
          \path [-] (x3) edge (w1);
          \path [-] (x4) edge (w1);
          \path [-] (x5) edge (w2);
          \path [-] (x6) edge (w2);
          \path [-] (x7) edge (w2);
          \path [-] (x4) edge (w2);
          \draw (x1) to [out=165,in=195,looseness=8] (x1);
          \path (x2) edge [loop above] node {} (x2);
          \path (x3) edge [loop below] node {} (x3);
          \path (x4) edge [loop above] node {} (x4);
          \path (x5) edge [loop above] node {} (x5);
          \path (x6) edge [loop below] node {} (x6);
          \draw (x7) to [out=15,in=-15,looseness=8] (x7);
    \end{scope}
    \end{tikzpicture}}
    \caption{Two clusters of nodes that join together to form a `butterfly' structure. Self-loops represent RPS self-play games, while edges between nodes represent $(I,-I)$. The \textit{red} nodes denote a population of species, while the \textit{blue} nodes stand for an environment.}
    \label{fig:butterfly1}
\end{figure}
\begin{figure}[t]
 \centering
\includegraphics[width=0.99\linewidth]{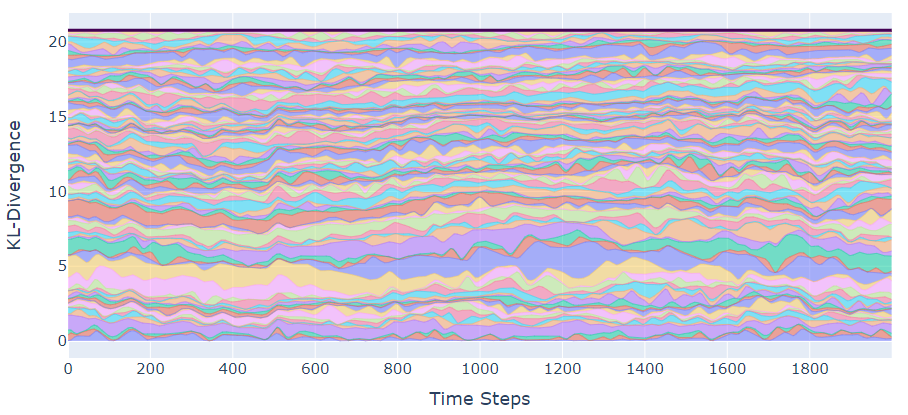}
 \caption{Weighted KL divergence for 25 cluster (100 player) time-evolving zero-sum game (for details see Appendix~\ref{appsecs:multiplayer}).}
 \label{fig:kldivtorus}
\end{figure}

The goal of this section is to experimentally verify some of the key results, and to highlight other empirically observed properties outside the established theoretical results.\footnote{Code is available at \href{https://github.com/ryanndelion/egt-squared}{github.com/ryanndelion/egt-squared}}

Theorem~\ref{t:main} states that any population/environment dynamics which can be captured via a \textit{rescaled zero-sum game} (no matter the complexity of such a description) exhibit a type of \emph{cyclic behavior} known as Poincar\'{e} recurrence. Indeed, the  trajectories shown in Figure~\ref{fig:main0} from the 
time-evolving generalized RPS game 
of Section~\ref{sec:evolution} are cyclic in nature.  
Specifically, Figure~\ref{fig:sub3} shows the coevolution of the system for a fixed initial condition. We plot the joint trajectory of the first two strategies for both the population $\species$ and environment $\weights$, which creates a 4D space where the color legend acts as the final dimension. 
In the supplementary code, we provide an animation of these dynamics for a range of initial conditions.
The simulation demonstrates that as the initial conditions move closer to the interior equilibrium, the trajectories themselves remain bounded within a smaller region around the equilibrium, which confirms the bounded regret property of the dynamics from Proposition~\ref{thm:regretbdd}.

 Lemma~\ref{t:invariant} shows that for any \textit{rescaled zero-sum game} there is a constant of motion, namely $\Phi(t)$. It is easy to see from the definition of $\Phi(t)$ that a weighted sum of KL-divergences between the strategy vectors produced by replicator dynamics and an interior Nash Equilibrium is also a constant of motion (see Corollary~\ref{cor:newconstant} in the Appendix). We simulated an extension to the game depicted in Figure~\ref{fig:butterfly1} in which many `butterfly' clusters are joined in a toroid shape.
Figure~\ref{fig:kldivtorus} depicts our claim: although each agent specific divergence 
term $\eta_i \mathrm{KL}(x^*_i||x_i(t))$  fluctuates, the weighted sum $\sum_{i \in V}\eta_i\mathrm{KL}(x^*_i||x_i(t))$
remains constant.

\begin{figure}[t]\centering
\begin{tabular}{cccc}
  \includegraphics[width=0.13\linewidth]{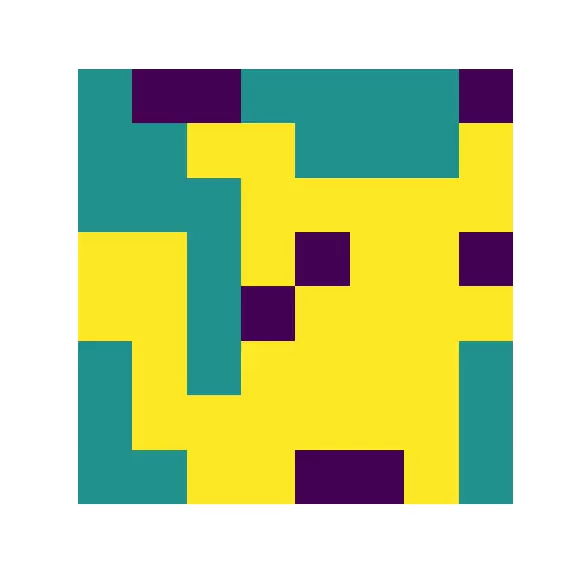} &   \includegraphics[width=0.13\linewidth]{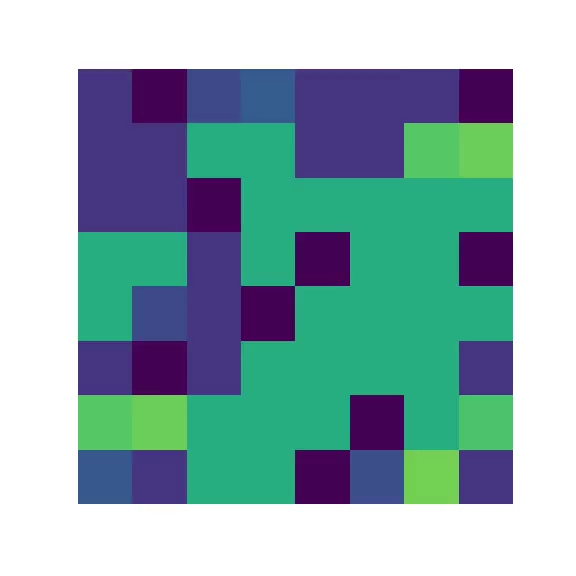} & \includegraphics[width=0.13\linewidth]{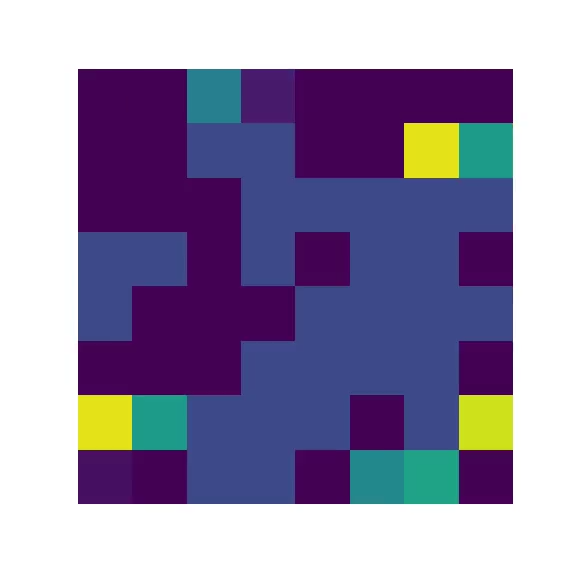} &   \includegraphics[width=0.13\linewidth]{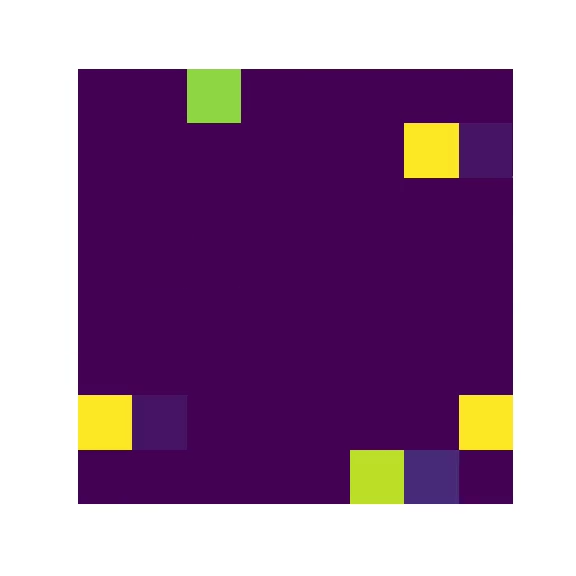} \\
(a) T=1 & (b) T=3 & (c) T=5 & (d) T=20 \\
 \includegraphics[width=0.13\linewidth]{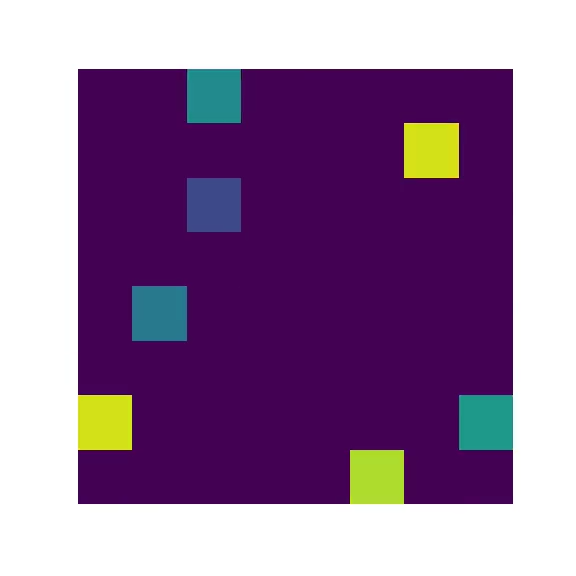} &   \includegraphics[width=0.13\linewidth]{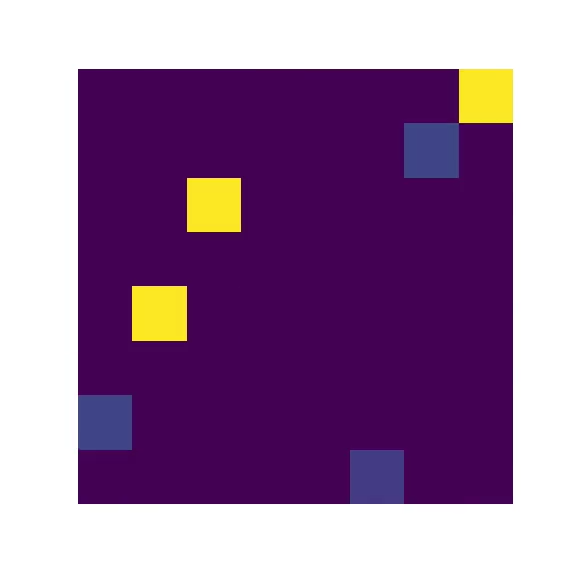} & \includegraphics[width=0.13\linewidth]{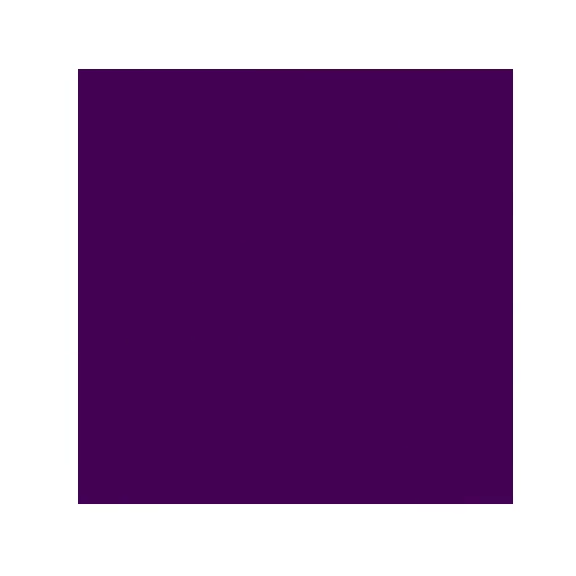} &   \includegraphics[width=0.13\linewidth]{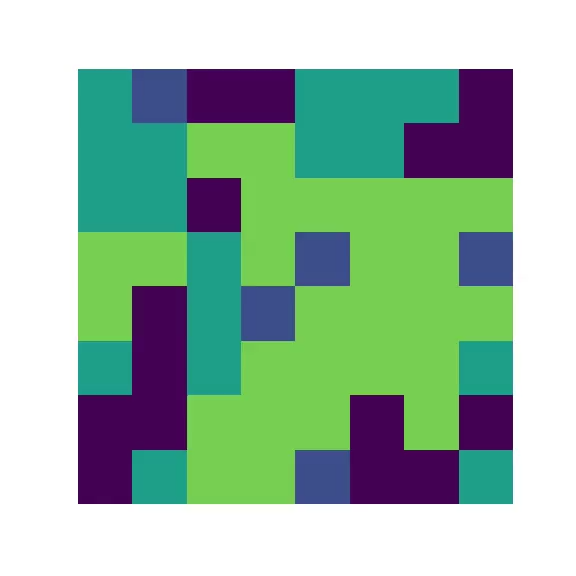} \\
 (e) T=47 & (f) T=54 & (g) T=68 & (h) T=93 \\
 \includegraphics[width=0.13\linewidth]{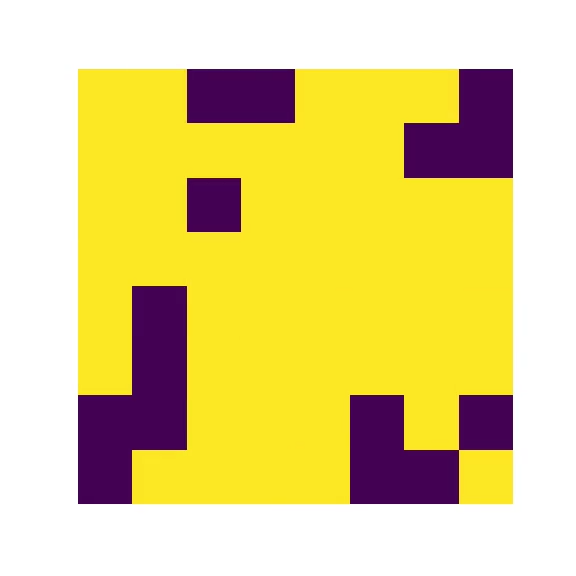} &   \includegraphics[width=0.13\linewidth]{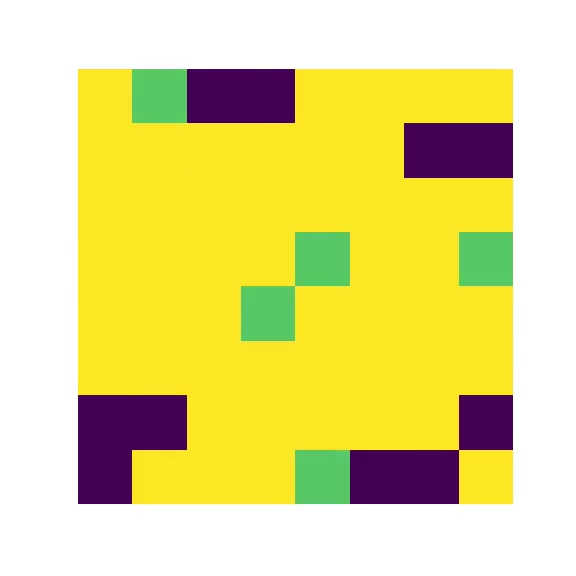} & \includegraphics[width=0.13\linewidth]{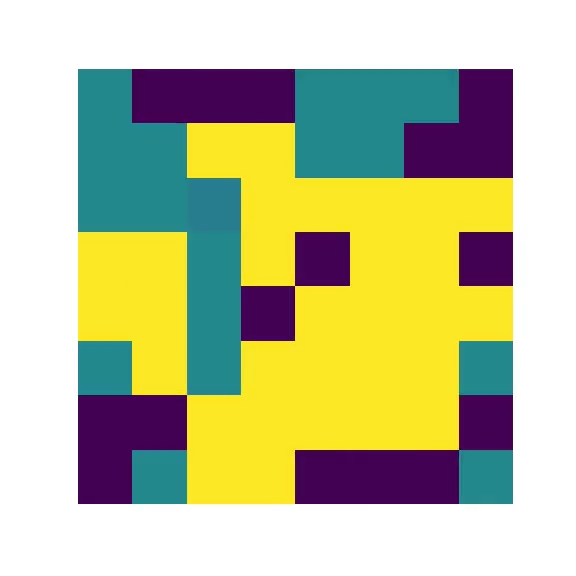} &   \includegraphics[width=0.13\linewidth]{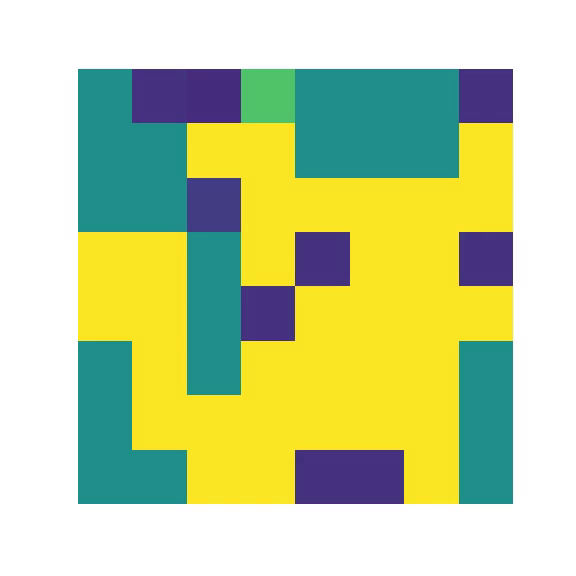} \\
 (i) T=99 & (j) T=112 & (k) T=117 & (l) T=101701 \\
\end{tabular}
\caption{Sequence of Pikachu images showing approximate recurrence in an $8\times 8$ zero-sum polymatrix game, where the changing color of each pixel on the grid represents the strategy of the player over time.}
\label{fig:pikachumain}
\end{figure}

To generate Figure~\ref{fig:pikachumain}, we scale-up the game structure from~\citet{mai2018cycles} to $64$ nodes. This is a relatively dense graph, where the initial condition of each player informs the RGB value of a corresponding pixel on a grid. If the system exhibits Poincar\'e recurrence, we should eventually see similar patterns emerge as the pixels change color over time (i.e., as their corresponding strategies evolve). 
In general, an upper bound on the expected time to see recurrence in such a system is exponential in the number of agents. As observed in Figure~\ref{fig:pikachumain}, the system returns near the initial image in the first several hundred iterations, but takes more than $100$k iterations for a clearer Pikachu to reappear. Further details on the simulation methodology and additional experiments can be found in Appendix~\ref{appsecs:experiments}. 
\FloatBarrier

\section{Discussion}
\label{sec:conclusion}

We show that systems in which populations of dynamic agents interact with environments that evolve as a function of the agents themselves can equivalently be modeled as polymatrix games. 
For the class of rescaled zero-sum games, we prove replicator dynamics are Poincar\'{e} recurrent and converge in time-average to the equilibrium, while experiments show the complexity of systems to which the results apply. A future direction of theoretical research is on the study of games that evolve exogenously instead of only endogenously.

Moreover, there are several exciting applications where our theory has relevance. Google DeepMind trains populations of AI agents against each other and computes win probabilities in heads-up competition resulting in a symmetric  constant-sum game~\cite{WGTSMJ20, balduzzi2019open}. Up to a shift by an all $0.5$ matrix, these are exactly anti-symmetric self-loop games connecting a population of users (programs) to itself as the programs are trying to out-compete each other.
The game always remains (anti)-symmetric, but the payoff entries change as stronger agents replace old agents. While we cannot capture the system fully, we can create the following abstract model of it. The self-loop zero-sum game is the initialization of the system and is equal to the original anti-symmetric empirical zero-sum game. There is another zero-sum game between the population and a meta-agent which simulates the reinforcement policy that chooses which programs get replaced and thus generates a new empirical zero-sum payoff matrix. We can mimic this randomized choice of the policy as a mixed strategy that chooses a convex combination from a  large number of  possible empirical zero-sum payoff matrices. One of these payoff matrices is the all zero matrix, and the initial strategy of the reinforcement policy chooses that game with high probability at time zero, so that the population is at the start of the process effectively playing just their original empirical game. For such systems, our results provide some theoretical justification for the preservation of diversity and for the satisfying empirical performance.

To conclude, we briefly touch on the connection to progressive training of generative adversarial networks~\cite{karras2018progressive}. The basic idea is to start the training process with small generator and discriminator networks and, over time, periodically add layers to the networks of higher dimension to grow the resolution of the generated images. This process causes the zero-sum game (between  generator and discriminator) to evolve with time. Importantly, as a consequence, the equilibrium is not fixed in the game. For instance, we can capture behavior of this process as a time evolving game in our model: the base game matrix $P$ is sparse and of high dimension; as the environment $w$ changes in time the nonzero values in the time-evolving payoff $P(w)$ `turn on', progressively making the matrix dense. Despite the critical nature of the above AI architectures, which are both based on the guided evolution of zero-sum games, no model of them exists in the literature.

\section*{Acknowledgments}
 Stratis Skoulakis gratefully acknowledges NRF 2018 Fellowship NRF-NRFF2018-07.
 Tanner Fiez acknowledges support from the DoD NDSEG Fellowship. 
Ryann Sim gratefully acknowledges support from the SUTD President's Graduate Fellowship (SUTD-PGF).
Lillian Ratliff is supported by NSF CAREER Award number 1844729 and an Office of Naval Research Young Investigor Award. 
Georgios Piliouras gratefully acknowledges support from  grant PIE-SGP-AI-2018-01, NRF2019-NRF-ANR095 ALIAS grant and NRF 2018 Fellowship NRF-NRFF2018-07.

\bibliographystyle{plainnat}
\bibliography{refs}

\renewcommand\appendixtocname{Appendix}
\renewcommand\appendixname{Appendix}
\begin{appendices}

\setcounter{secnumdepth}{3}

We provide a detailed overview of related work in Appendix~\ref{appsecs:related}, omitted proofs in Appendix~\ref{appsecs:proofs}, and further experimental results and details in Appendix~\ref{appsecs:experiments}. 

\section{Related Work}
\label{appsecs:related}
We now cover a broader class of related work. The related work can be categorized into the following topics: (i) learning in zero-sum games, Poincar\'{e} recurrence, and cycles, (ii) learning in time-evolving games, and (iii) experimental works.
\paragraph{Learning in Zero-Sum Games, Poincar\'{e} Recurrence, and Cycles.} 
 Classical works in evolutionary game theory have long explored the interface between dynamical systems theory and learning in games with the goal of understanding when cycling or other non-convergent behaviors emerge \cite{hofbauer1998evolutionary,sandholm2010population}. For specific classes of games such as zero-sum or partnership bimatrix games, constants of motion are known to exist \cite{hofbauer1996evolutionary}, and volume preservation properties of the replicator dynamics have been shown \cite{eshel1983coevolutionary,hofbauer1998evolutionary}. 
 More recently, applications of dynamical systems tools to the analysis of learning algorithms has lead to new insights about non-convergent behavior and its interpretation \citep{vlatakis2019poincare,boone2019darwin,piliouras2014optimization,piliouras2014persistent,mertikopoulos2018cycles,perolat2020poincar}. For instance, several works demonstrate the surprising property that replicator dynamics are Poincar\'{e} recurrent in pairwise zero-sum polymatrix games without self-loops by showing both the existence of a constant of motion and volume preservation~\cite{piliouras2014optimization,piliouras2014persistent}. \citet{boone2019darwin} extend this analysis to pairwise zero-sum polymatrix games with self loops. \citet{mai2018cycles} consider a biologically-inspired time-evolving game in which the payoff of a collection of species playing against itself depends on a dynamically changing environmental variable.  The dynamics are shown to be Poincar\'{e} recurrent  for certain parameter regimes, which is interpreted as promoting diversity.

\paragraph{Learning in Time-Evolving Games.} 
Following a similar theme, there has been renewed interest in learning in games in which the payoffs change in time or are affected by a feedback mechanism from the environment. Such reciprocal feedback between strategies and environment variables arises in a number of applications including biology \cite{akccay2011evolution,cortez2010understanding,tilman2020evolutionary}, ecology \cite{worden2007evolutionary,lade2013regime}, sociology~\cite{bowles2003co,tilman2017maintaining}, economics \cite{friedman1991evolutionary}, and more recently machine learning and artificial intelligence~\cite{cardoso2019competing,duvocelle2018learning,lykouris2016learning}. For instance, \citet{cardoso2019competing} design algorithms with \emph{small regret}---tantamount to the long-term payoff of both players being close to minimax optimal in hindsight---for a class of repeated play  zero-sum games, termed \emph{online matrix games}, such that   players' payoff matrices may change in each round.
In related work, in the class of continuous games, \citet{duvocelle2018learning} analyze the long-run behavior of regret minimizing players in time-evolving games which are executed in a sequence of concave, monotone \emph{stage games}. In other work~\cite{lykouris2016learning,cavaliere2012prosperity}, dynamically changing environments are modeled via a dynamic player population in which players leave the game with some probability and new players enter.

Closer to the class of time-evolving games we consider, 
another body of work captures various natural dynamical processes via action-dependent games \cite{SAGD06,SBGA07,SB02,KB10,RDK14,BGE03,BE04,GRC07,weitz2016oscillating,mai2018cycles,akccay2011evolution,stewart2014collapse,mullon2018social}. In such settings,
the actions of the players (or in terms of evolutionary game-theory, the frequencies of species within a population), may affect the environment and thus change the game's payoffs. For instance, the dynamics of the vaccinated human population are efficiently captured by such action dependent-games. Parents decide to vaccinate their newborns by weighing the risk of a potential disease to the risk of morbidity of vaccination. However, as the unvaccinated population increases so does the cost of the \textit{do not vaccinate action} \cite{KB10,RDK14,BGE03,BE04,GRC07}. Similar instances appear in the dynamics of bacteria and microbe populations since it is common for certain types of bacteria to cause certain environmental changes (e.g., increase of nutrient-scavenging enzymes or fixation of inorganic nutrients) that affect differently
the various individuals of the population \cite{SAGD06,SBGA07,SB02}.

\paragraph{Experimental Works.}
Motivated by the observation of cycling behavior in applications of game-theoretic learning dynamics to machine learning dynamics, there has been a push to better understand recurrence and to potentially see it  as a solution concept. Indeed, standard gradient dynamics are known to result in cycling or recurrent behavior in continuous time~\cite{piliouras2014optimization, mertikopoulos2018cycles} and chaotic and divergent behavior in discrete time~\cite{bailey2018multiplicative,cheung2019vortices,gidel2019negative}.  \citet{daskalakis2018training} and \citet{mertikopoulos2018cycles} explore adaptations to follow-the-regularized learning (FTRL) dynamics that enable convergence in bilinear and more general nonconvex-nonconcave problems, respectively.
Each work shows that versions of optimistic mirror descent can successfully train generative adversarial networks on challenging datasets.
Similarly, \citet{balduzzi2018mechanics} design dynamics that adjust for components of the gradient dynamics that cause cycling by drawing connections to Hamiltonian dynamics. 

As opposed to trying to mitigate cycling behavior and converge to fixed points via carefully designed learning dynamics, a separate line of work instead makes use of the fact that in convex-concave games the time-average of standard gradient dynamics converge to the interior equilibrium~\cite{freund1999adaptive}
. In particular,~\citet{ganaveraging} show that training generative adversarial networks and then averaging the parameters of the networks uniformly or by an exponential moving average is an empirically effective method. \citet{gidel2019negative} explore a similar perspective of uniform averaging for simultaneous and alternating gradient updates using negative momentum. Moreover,~\citet{vlatakis2019poincare} show for a class of nonconvex-nonconcave minimax games, which generalize bilinear zero-sum games, that the time-average of gradient dynamics  converges to an equilibrium for certain problem instances and initial conditions. This provides further evidence for the efficacy of recurrence as a solution concept that is relevant to machine learning applications such as generative adversarial networks. A final line of work explores evolutionary algorithms as a training method for generative adversarial networks~\cite{costa2020using, karras2018progressive, wang2019evolutionary, costa2019coegan, costa2019coevolution, garciarena2018evolved, al2018towards, toutouh2019spatial}.

\section{Proofs}
\label{appsecs:proofs}
We organize the proofs in the order that the results appeared in the paper. Proofs for results from Sections~\ref{sec:evolution},~\ref{sec:recurrence}, and~\ref{sec:timeaverage} of the paper can be found in Appendix~\ref{appsecs:evolution_proof},~\ref{appsecs:recurrence_proof},~\ref{appsecs:regret}, respectively.

In Appendix~\ref{appsecs:evolution_proof}, we begin by proving Proposition~\ref{prop:rps}, which shows a reduction from the time-evolving generalized rock-paper-scissors game presented by~\citet{mai2018cycles} to a rescaled zero-sum polymatrix game. Following that proof, we prove Theorem~\ref{thm:evolution_general}, which shows the reduction of Proposition~\ref{prop:rps} extends to a general class of time-evolving dynamical systems. This result demonstrates the breadth of time-evolving games that can in fact be studied as rescaled zero-sum polymatrix games.

In Appendix~\ref{appsecs:recurrence_proof}, we prove Lemma~\ref{l:divergence} and Lemma~\ref{t:invariant}. Recall that Lemma~\ref{l:divergence} and Lemma~\ref{t:invariant} show that the replicator dynamics are volume preserving and have bounded orbits in rescaled zero-sum polymatrix games, respectively. Moreover, as shown in Section~\ref{sec:recurrence} via the proof of Theorem~\ref{t:main}, Lemma~\ref{l:divergence} and Lemma~\ref{t:invariant} nearly immediately imply Theorem~\ref{t:main}, which guarantees the replicator dynamics are Poincar\'{e} recurrent in rescaled zero-sum polymatrix games with interior Nash equilibria.

Appendix~\ref{appsecs:regret} contains the proofs of Theorem~\ref{t:avg1} and Proposition~\ref{thm:regretbdd}, which show time-average equilibria and utility convergence of replicator dynamics in rescaled zero-sum polymatrix games along with the bounded regret property in general polymatrix games, respectively. The proof of Theorem~\ref{thm:lp}, which provides a linear program to compute Nash equilibrium in rescaled zero-sum polymatrix games can also be found in Appendix~\ref{appsecs:regret}. 

\subsection{Proofs of Reductions from Time-Varying Games to Polymatrix Games from Section~\ref{sec:evolution}}
\label{appsecs:evolution_proof}
In Appendix~\ref{appsecs:reduction_proof}, we provide the proof of Proposition~\ref{prop:rps} from Section~\ref{sec:evolution}. This result shows a reduction from a time-evolving generalized rock-paper-scissors game to an appropriate rescaled zero-sum polymatrix game. Moreover, in Appendix~\ref{appsecs:generalization}, we provide the proof of Theorem~\ref{thm:evolution_general} from Section~\ref{sec:evolution}, which generalizes the reduction of Proposition~\ref{prop:rps} to a broad class of dynamical systems that represent multiple evolving populations interacting with multiple evolving environments.
\subsubsection{Proof of Proposition~\ref{prop:rps}}
\label{appsecs:reduction_proof}
Let $\species$ and $\weights$ denote the mixed strategies of player 1 and player 2, respectively, which correspond to the population and the environment. In what follows, we show that both the population and environment dynamics can be simplified so that it is clear each player is following replicator dynamics in a static rescaled zero-sum polymatrix game. 

\paragraph{Environment Dynamics.}
We begin by considering the environment dynamics. The dynamics of player 2 ($\weights$-player) for each action $i \in \{1,\dots, n\}$ with an initial condition on the interior of the simplex are given by
\begin{equation}
\dot{\weights}_i = \weights_i \sum_{j=1}^n\weights_j\left(\species_j - \species_i\right).
\label{eq:repeat_w}
\end{equation}
Now observe that
\begin{align*}
\sum_{i=1}^n\dot{\weights}_i 
&= \sum_{i=1}^n\weights_i\sum_{j=1}^n\weights_j\left(\species_j - \species_i\right) \\
&= \sum_{i=1}^n\weights_i\sum_{j=1}^n\weights_j\species_j - \sum_{i=1}^n\weights_i\sum_{j=1}^n\weights_j\species_i \\
&= \sum_{i=1}^n\weights_i\sum_{j=1}^n\weights_j\species_j - \sum_{j=1}^n\weights_j\sum_{i=1}^n\weights_i\species_i \\
&= 0.
\end{align*}
Since $\sum_{i=1}^n\dot{\weights}_i =0$ as shown above and the given initial condition is such that $\weights(0) \in \Delta^{n-1}$, we conclude ${\weights}(t) \in \Delta^{n-1}$ and $\sum_{j=1}^n\weights_j(t) = 1$ for any $t\geq 0$. 
From a series of algebraic manipulations and the fact that $\sum_{j=1}^n\weights_j = 1$, we obtain an equivalent form of the dynamics given in~\eqref{eq:repeat_w} for each action $i \in \{1,\dots, n\}$ as follows:
\begin{align}
\dot{\weights}_i 
&= \weights_i \sum_{j=1}^n\weights_j\left(\species_j - \species_i\right) \notag  \\
&= \weights_i \sum_{j=1}^n\weights_j\species_j - \weights_i \sum_{j=1}^n\weights_j\species_i \notag \\
&= \weights_i \sum_{j=1}^n\weights_j\species_j - \weights_i\species_i \notag \\
&= \weights_i \big(- \species_i +  \sum_{j=1}^n\weights_j \species_j\big).\label{eq:w_final}
\end{align}
The dynamics for player 2 ($\weights$-player) from~\eqref{eq:w_final} in vector form are then given by
\begin{equation}
\dot{\weights} = \weights\cdot (-\I \species +  \weights^{\top} \I \species).
\label{eq:w_dynam_vec}
\end{equation}
We now see that the dynamics in~\eqref{eq:w_dynam_vec} are replicator dynamics in which player 2 ($w$-player) plays against player 1 ($y$-player) with the payoff matrix $A^{w,y} = -\I$, where the superscript indices $(w,y)$ indicate the players.

\paragraph{Population Dynamics.}
We now perform a 
similar analysis on the population dynamics. The dynamics for player 1 ($\species$-player) for each action $i \in \{1,\dots,n\}$ with an initial condition on the interior of the simplex are given by
\begin{equation}
\dot{\species}_i = \species_i  \left(
(\P(\weights) \species)_i - \species^\top \P(\weights) \species
\right).
\label{eq:repeat_y}
\end{equation}
From an expansion of the payoff matrix $\P(\weights)$ in~\eqref{eq:repeat_y}, the dynamics of player 1 ($\species$-player) for each action $\{1,\dots, n\}$ are equivalently
\begin{equation}
\dot{\species}_i = \species_i  \left( 
(\P \species)_i - \species^\top 
\P \species
\right) + \species_i  \Bigg(
\mu \sum_{j=1}^n(\weights_i - \weights_j)\species_j-\mu\sum_{\ell=1}^n \species_\ell \sum_{j=1}^n(\weights_\ell - \weights_j) \species_j
\Bigg).
\label{eq:y_expanded}
\end{equation}
Observe that
\begin{align*}
\sum_{\ell=1}^n \species_\ell \sum_{j=1}^n(\weights_\ell - \weights_j) \species_j
&=\sum_{\ell=1}^n \species_\ell \sum_{j=1}^n\weights_\ell\species_j - \sum_{\ell=1}^n\species_\ell \sum_{j=1}^n\weights_j \species_j  \\
&=\sum_{j=1}^n\species_j\sum_{\ell=1}^n \weights_\ell\species_\ell - \sum_{\ell=1}^n\species_\ell \sum_{j=1}^n\weights_j \species_j \\
&=0.
\end{align*}
Consequently, for each action $i \in \{1,\dots, n\}$, the dynamics in~\eqref{eq:y_expanded}  simplify to the form
\begin{equation}
\dot{\species}_i = \species_i  \left( 
(\P \species)_i - \species^\top \P \species
\right) +  \species_i  \left(
\mu \sum_{j=1}^n(\weights_i - \weights_j)\species_j\right).
\label{eq:y_inter}
\end{equation}
Finally, $\species(0) \in \Delta^{n-1}$ so that $\sum_{j=1}^n\species_j(t) = 1$ for any $t\geq 0$ since clearly $\dot{\species}$ is replicator dynamics with the payoff matrix $\P(\weights)$ in~\eqref{eq:repeat_y}. Accordingly, for each action $i \in \{1,\dots, n\}$, we simplify the dynamics in~\eqref{eq:y_inter} as follows:
\begin{align}
\dot{\species}_i &= \species_i  \left( 
(\P \species)_i - \species^\top \P \species
\right) +  \species_i  \big(
\mu \sum_{j}(\weights_i - \weights_j)\species_j\big) \notag \\
 &= \species_i  \left( 
(\P \species)_i - \species^\top \P \species
\right) +  \species_i  \big(
\mu\weights_i \sum_{j=1}^n\species_j - \sum_{j=1}^n\weights_j\species_j\big) \notag \\
 &= \species_i  \left( 
(\P \species)_i - \species^\top \P \species
\right) + \species_i \big(
\mu \weights_i - \sum_{j=1}^n \mu \weights_j \species_j
\big). \label{eq:y_final}
\end{align}
The dynamics for player 1 ($\species$-player) from~\eqref{eq:y_final} in vector form are then given by
\begin{equation}
\dot{\species} = \species\cdot (\P \species +  \species^{\top} \P \species\cdot \ones) + \species\cdot (\mu\I \weights +  \mu \species^{\top} \I \weights\cdot \ones).
\label{eq:y_dynam_vec}
\end{equation}
We now see that the dynamics in \eqref{eq:y_dynam_vec} are replicator dynamics in which player 1 ($\species$-player) plays against itself with the payoff matrix $A^{y,y} = \P$ and against player 2 ($w$-player) with the payoff matrix $A^{y,w} = \mu \I$, where again the superscript indices indicate the players in the payoff matrix.

\paragraph{Static Polymatrix Game.}
 We have shown that the dynamics of~\eqref{eq:repeat_y} and~\eqref{eq:repeat_w} correspond to replicator dynamics for a two-player polymatrix game in which player 1 ($\species$-player) has utility $u_{\species}(\species, \weights)=\species^\top \P \species +\mu \species^\top \I \weights$ and player 2 ($\weights$-player) has utility $u_{\weights}(\species, \weights)=-\weights^\top \I \species$ for any strategy profile $(\species, \weights)$. The self-loop of player 1
 ($\species$-player) is antisymmetric and for $\sconst_{\species}=1$ and $\sconst_{\weights} =\mu$ the rescaled sum of utility $\sconst_{\species} u_{\species}(\species, \weights) + \sconst_{\weights} u_{\weights}(\species, \weights)=0$ for every strategy $(\species, \weights)$. This allows us to conclude by definition that the time-evolving generalized rock-paper-scissors game is equivalent to replicator dynamics in a two-player rescaled zero-sum polymatrix game.

\subsubsection{Proof of Theorem~\ref{thm:evolution_general}}
\label{appsecs:generalization}
In this section, we provide the proof of Theorem~\ref{thm:evolution_general}. The theorem 
shows that the reduction from Proposition~\ref{prop:rps} extends to more general dynamical systems.
Before giving the proof, we provide some intuition for the underlying structure. 
\paragraph{Time-Evolving Games that Admit Reduction to Polymatrix Games.}
\label{appsecs:structure}
The class of time-evolving systems that admit a reduction are such that the basic interaction structure is of the form in Figure~\ref{fig:buildingblock}. 
The key component of any general structure formed from the building block is that each environment $w_k$ is only connected to populations, and each population $y_\ell$ is only connected to environments or themselves via a self-loop. 
As an example of the type of generalization that is possible for the reduction, consider the polymatrix game in Figure~\ref{fig:graphicalgame}. Of course the game graph does not have to be a line, but population nodes should be separated by environment nodes.

\begin{figure}[t]
    \centering
    \begin{tikzpicture}
    \begin{scope}[every node/.style={circle,thick,draw}]
        \node[fill =red!30] (x2) at (0,2) {$\species_\ell$};
        \node[fill =blue!30] (w1) at (3,2) {$\weights_k$};
    \end{scope}
    \begin{scope}[>={Stealth[black]},
                 every node/.style={fill=white,circle},
                 every edge/.style={draw=black}]
        
        \path [->] (x2) edge [bend right=30] (w1);
        \path [->] (w1) edge [bend right=30]  (x2);
        \path (x2) edge [loop left]  (x2);
    \end{scope}
    \end{tikzpicture}
    \caption{Basic interaction structure in the time-evolving systems that reduce to rescaled zero-sum polymatrix games. The \textit{red} nodes denotes a population of species, while the \textit{blue} node is an environment. }
    \label{fig:buildingblock}
\end{figure}
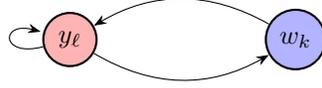
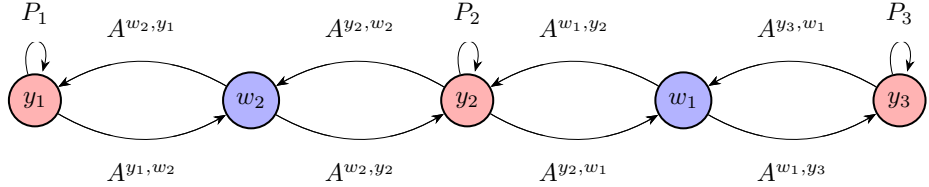
\begin{figure}[t]
    \centering
    \resizebox{0.75\linewidth}{!}{
    \begin{tikzpicture}
    \begin{scope}[every node/.style={circle,thick,draw}]
        \node[fill =red!30] (x1) at (-6,2) {$\species_1$};
        \node[fill =red!30] (x2) at (0,2) {$\species_2$};
        \node[fill =red!30] (x3) at (6,2) {$\species_3$};
        
        \node[fill =blue!30] (w1) at (3,2) {$\weights_1$};
        \node[fill =blue!30] (w2) at (-3,2) {$\weights_2$};
    \end{scope}
    \begin{scope}[>={Stealth[black]},
                 every node/.style={fill=white,circle},
                 every edge/.style={draw=black}]
        \path [->] (x1) edge [bend right=30] node [below=-0.2cm] {$A^{y_1,w_2}$}(w2);
          \path [->] (x1) edge [bend right=30] (w2);
        
        \path [->] (w2) edge [bend right=30] node [above=-0.2cm] {$A^{w_2,y_1}$} (x1);
                \path [->] (w2) edge [bend right=30] (x1);
        \path (x1) edge [loop above] node {$\P_1$} (x1);
        
        \path [->] (x2) edge [bend right=30] node [below=-0.2cm] {$A^{y_2,w_1}$}(w1);
        \path [->] (x2) edge [bend right=30] (w1);
        
        \path [->] (w1) edge [bend right=30] node [above=-0.2cm] {$A^{w_1,y_2}$}(x2);
          \path [->] (w1) edge [bend right=30](x2);
        
        \path (x2) edge [loop above] node {$\P_2$} (x2);
        
        \path [->] (x2) edge [bend right=30] node [above=-0.2cm] {$ A^{y_2,w_2}$}(w2);
         \path [->] (x2) edge [bend right=30] (w2);

        \path [->] (w2) edge [bend right=30] node [below=-0.2cm] {$A^{w_2,y_2}$} (x2);
     \path [->] (w2) edge [bend right=30]  (x2);
     
        \path [->] (x3) edge [bend right=30] node [above=-0.2cm] {$A^{y_3,w_1}$}(w1);
        
            \path [->] (x3) edge [bend right=30] (w1);
        
        \path [->] (w1) edge [bend right=30] node [below=-0.2cm] {$A^{w_1,y_3}$} (x3);
          \path [->] (w1) edge [bend right=30] (x3);
           
        \path (x3) edge [loop above] node {$\P_3$} (x3);
    \end{scope}
    \end{tikzpicture}}
    \caption{Example polymatrix game that can be formed from a reduction of a time-evolving dynamical system. The \textit{red} nodes denote a population of species, while the \textit{blue} nodes stand for an environment.}
    \label{fig:graphicalgame}
\end{figure}
Formally, a time-evolving game is defined by a set of populations $(y_1,\ldots, y_{n_y})$ and a set of environments $(w_1,\ldots, w_{n_w})$, where $y_\ell(0)\in \Delta^{n-1}$ for each $\ell\in\{1,\ldots, n_y\}$ and $w_k(0)\in \Delta^{n-1}$ for each $k\in \{1,\ldots, n_w\}$. Let $\mc{N}_k^w$ be the set of populations which coevolve with the environment $w_k$ and $\mc{N}_\ell^y$ be the set of environments which coevolve with the population $y_\ell$ via the building block structure from Figure~\ref{fig:buildingblock}. The time-evolving dynamics for each population $\ell$ are given componentwise by
\begin{equation}
\dot{\species}_{\ell,i} = \species_{\ell,i}  \left(
(\P_\ell({w}) \species_\ell)_{i} - \species_\ell^\top \P_\ell(\weights) \species_\ell 
\right), 
\label{eq:fbevolutiony}
\end{equation}
where 
\[\P_\ell(w)=\P_\ell+\sum_{k\in \N^y_\ell}\W^{\ell,k}\]
and $\W^{\ell,k}$ is a matrix such that the $(r,s)$ entry is given by 
\[W^{\ell,k}=\begin{pmatrix}
0 & (A^{\ell,k}w_{k})_1 - (A^{\ell,k}w_{k})_2 &  \cdots &  (A^{\ell,k}w_{k})_1 - (A^{\ell,k}w_{k})_n\\
(A^{\ell,k}w_{k})_2 - (A^{\ell,k}w_{k})_{1} & 0 & \cdots & (A^{\ell,k}w_{k})_{2}-(Aw_{k})_{n}\\
\vdots & \vdots  & \vdots & \vdots\\
(A^{\ell,k}w_{k})_{n} - (A^{\ell,k}w_{k})_{1} &(A^{\ell,k}w_{k})_{n}-(A^{\ell,k}w_{k})_{2} & \cdots & 0\\
\end{pmatrix}.\]

Further, the time-evolving dynamics for each environment $k$ are given componentwise by
\begin{equation}
\dot{\weights}_{k,i} = \weights_{k,i} \sum_{\ell\in \N^w_k}\sum_{j=1}^n\weights_{k,j}\left(  (A^{k,\ell}\species_{\ell})_{i}-(A^{k,\ell}\species_{\ell})_{j}\right).
\label{eq:fbevolutionw}
\end{equation}
We now prove that any time-evolving game defined by the dynamics in~(\ref{eq:fbevolutiony}--\ref{eq:fbevolutionw}) is equivalent to replicator dynamics in a polymatrix game.

\paragraph{Environment Dynamics.} 
We begin by showing that the dynamics for each environment reduces to replicator dynamics for a polymatrix game in which each environment plays edge games with each of the populations to which it is connected. 

Following a similar argument as in the proof of Proposition~\ref{prop:rps}, for each $k\in \{1,\ldots, n_w\}$, given that $w_k(0)\in \Delta^{n-1}$, we have that $w_k(t)\in \Delta^{n-1}$ for all $t\geq 0$. 
Since $\sum_{j=1}^n\weights_{k,j}(t) = 1$  for any fixed $t$ and for each environment $k$, an equivalent form of the dynamics given in~\eqref{eq:fbevolutionw} is
\begin{align*}
\dot{\weights}_{k,i} &= \weights_{k,i} \sum_{\ell\in \N^w_k}\sum_{j=1}^n\weights_{k,j}\big(  (A^{k,\ell}\species_{\ell})_{i}-(A^{k,\ell}\species_{\ell})_{j}\big)\\
 &= \weights_{k,i} \sum_{\ell\in \N^w_k}\Big(\sum_{j=1}^n\weights_{k,j}(A^{k,\ell}\species_{\ell})_{i}-\sum_{j=1}^n\weights_{k, j}(A^{k,\ell}\species_{\ell})_{j}\Big)\\
&= \weights_{k,i} \sum_{\ell\in \N^w_k}\Big((A^{k,\ell}\species_{\ell})_{i}-\sum_{j=1}^n\weights_{k,j}(A^{k,\ell}\species_{\ell})_{j}\Big) \\
&= \weights_{k,i} \sum_{\ell\in \N^w_k}\big((A^{k,\ell}\species_{\ell})_{i}-\weights_{k}^{\top}A^{k,\ell}\species_{\ell}\big).
\end{align*}
It is now clear that the dynamics from~\eqref{eq:fbevolutionw} equivalently correspond to replicator dynamics where each environment $k$
plays against each connected population $\ell\in \N_k^w$ with payoff matrix $A^{k,\ell}$.

\paragraph{Population Dynamics.} We now show that the dynamics for each of the populations reduce to replicator dynamics for a population playing against themselves in a self-loop game and against the environments to which the population is connected.

To begin, from an expansion of the payoff matrix $\P_\ell(\weights)$, the dynamics from~\eqref{eq:fbevolutiony} are equivalent to
\begin{equation*}
\begin{split}
\dot{\species}_{\ell,i}&= \species_{\ell,i}  \left( 
(\P_\ell \species_\ell)_i - \species_\ell^\top 
\P_\ell \species_\ell
\right) \\ &\quad+ \species_{\ell,i}  \Big(
\sum_{k\in \N_\ell^y} \sum_{j=1}^n((A^{\ell,k}\weights_{k})_{i} - (A^{\ell,k}\weights_{k})_{j})\species_{\ell,j}- \sum_{k\in \N_\ell^y}\sum_{s=1}^n \species_{\ell,s} \sum_{j=1}^n((A^{\ell,k}\weights_{k})_{s} - (A^{\ell,k}\weights_{k})_{j}) \species_{\ell,j}
\Big).
\end{split}
\label{eq:y_expanded_long}
\end{equation*}
Now, observe that for each $k\in \N_\ell^y$,
\begin{align*}
\sum_{s=1}^n \species_{\ell,s} \sum_{j=1}^n((A^{\ell,k}\weights_{k})_{s} - (A^{\ell,k}\weights_{k})_{j}) \species_{\ell,j} &=\sum_{s=1}^n \species_{\ell,s} \sum_{j=1}^n(A^{\ell,k}\weights_{k})_{s}\species_{\ell,j}\mkern-2mu  - \mkern-3mu\sum_{s=1}^n\species_{\ell,s} \sum_{j=1}^n(A^{\ell,k}\weights_{k})_{j} \species_{\ell,j}  \\
&=\sum_{j=1}^n\species_{\ell,j}\sum_{s=1}^n (A^{\ell,k}\weights_{k})_{s}\species_{\ell,s}\mkern-2mu - \mkern-3mu\sum_{s=1}^n\species_{\ell,s} \sum_{j=1}^n(A^{\ell,k}\weights_{k})_{j} \species_{\ell,j} \\
&=0.
\end{align*}
Hence, along with the fact that $\sum_{j=1}^ny_{\ell,j}=1$, we obtain
\begin{align*}
\dot{\species}_{\ell,i} &= \species_{\ell,i} \left( (\P_\ell \species_\ell)_i - \species_\ell^\top \P_\ell \species_\ell\right) + \species_{\ell,i}  \sum_{k\in \N_\ell^y} \sum_{j=1}^n((A^{\ell,k}\weights_{k})_{i} - (A^{\ell,k}\weights_{k})_{j})\species_{\ell,j}\\
&=\species_{\ell,i}  \left( 
(\P_\ell \species_\ell)_i - \species_\ell^\top 
\P_\ell \species_\ell
\right) + \species_{\ell,i} 
\sum_{k\in \N_\ell^y} \Big(\sum_{j=1}^n (A^{\ell,k}\weights_{k})_{i}\species_{\ell, j} -\sum_{j=1}^n (A^{\ell,k}\weights_{k})_{j}\species_{\ell,j}
\Big)\\
&=\species_{\ell,i}  \left( 
(\P_\ell \species_\ell)_i - \species_\ell^\top \P_\ell \species_\ell
\right)+ \species_{\ell,i} \sum_{k\in \N_\ell^y} \Big((A^{\ell,k}\weights_{k})_{i} -\sum_{j=1}^n (A^{\ell,k}\weights_{k})_{j}\species_{\ell,j}\Big) \\
&=\species_{\ell,i}  \left( 
(\P_\ell \species_\ell)_i - \species_\ell^\top \P_\ell \species_\ell
\right)+ \species_{\ell,i} \sum_{k\in \N_\ell^y} \big((A^{\ell,k}\weights_{k})_{i} -\species_{\ell}^{\top}A^{\ell,k}\weights_{k}\big).
\end{align*}
The final equation shows that the dynamics from~\eqref{eq:fbevolutiony} equivalently correspond to replicator where each population $\ell$ plays against itself with the payoff matrix $A^{\ell, \ell} = \P_{\ell}$ and against each environment $k\in \N_{\ell}^y$ to which it is connected with payoff matrix $A^{\ell,k}$.

\paragraph{Static Polymatrix Game.} It is now clear that the dynamics from (\ref{eq:fbevolutiony}--\ref{eq:fbevolutionw}) correspond to  replicator dynamics for a polymatrix game with $V$ the combined index set of environments and populations such that $|V|=n_y+n_w$.
The edge games are defined such that each population player $\ell$ plays against themselves with $A^{\ell,\ell}=\P_\ell$ and against each environment $k\in \N_\ell^w$ to which they are connected with game $A^{\ell,k}$, and such that each environment $k$ plays against each population $\ell\in \N_k^w$ to which it is connected with $A^{k,\ell}$. 
If each $P_\ell=-P_\ell^\top$ and $\sum_{i\in V}\eta_iu_i(x)=0$ for all $x \in \X$ and some positive coefficients $\{\eta_i\}_{i\in V}$, then the polymatrix game is rescaled zero-sum. 

Finally, we remark that while it may appear complex to verify if the polymatrix game resulting from the reduction of the time-evolving dynamics is rescaled zero-sum, \citet{CCDP16} have shown that whether a polymatrix game is constant-sum can be determined in polynomial time and this result can apply to rescaled zero-sum games.
\begin{theorem}[Theorem 8~\cite{CCDP16}]
Let $G= (V, E)$ be a polymatrix game. For any player $i \in V$, pure strategy $\alpha \in \A_i$, and joint strategy $x_{-i}$ of the rest of the players, denote by $W(\alpha, x_{-i})=\sum_{j\in V}u_j(\alpha, x_{i-1})$ the sum of all players' payoffs when agent $i$ plays strategy $\alpha$ and the rest of the agents play $x_{-i}$. The polymatrix game $G$ is a constant-sum game if and only if the optimal objective value of the problem
\begin{equation*}
\max_{x_{-i}} \quad W(\beta, x_{-i})-W(\alpha, x_{-i})  
\end{equation*}
equals zero for all $i \in V$ and $\alpha, \beta \in \A_i$. Moreover, this condition can be checked in polynomial time in the number of players and strategies.
\end{theorem}

\subsection{Proof of Poincar\'e Recurrence in Rescaled Zero-Sum Polymatrix Games from Section~\ref{sec:recurrence}}
\label{appsecs:recurrence_proof}
We now provide the proofs of Lemma~\ref{l:divergence} and Lemma~\ref{t:invariant} from Section~\ref{sec:recurrence} in Appendix~\ref{appsecs:divergence_proof} and Appendix~\ref{appsecs:invariant_proof}, respectively. Recall that Lemma~\ref{l:divergence} and Lemma~\ref{t:invariant} show that the replicator dynamics are volume preserving and have bounded orbits in rescaled zero-sum polymatrix games, respectively. Moreover, as shown in Section~\ref{sec:recurrence} via the proof of Theorem~\ref{t:main}, Lemma~\ref{l:divergence} and Lemma~\ref{t:invariant} nearly immediately imply Theorem~\ref{t:main}, which guarantees the replicator dynamics are Poincar\'{e} recurrent in rescaled zero-sum polymatrix games with interior Nash equilibria.
\subsubsection{Proof of Lemma~\ref{l:divergence}}
\label{appsecs:divergence_proof}
We need to show 
\begin{equation*}
\sum_{i \in V}\sum_{\alpha \in \A_i}\frac{d }{dz_{i\alpha}}F_{i\alpha}(z)=0.
\end{equation*}
Recall from~\eqref{eq:zdynamics} that for each $\alpha \in \A_i$ and $i \in V$,
\begin{equation*}
F_{i\alpha}(z)
= \sum_{j \in V} \sum_{\beta \in \A_j} A^{ij}_{\alpha \beta}  \frac{e^{z_{j \beta}}}{\sum_{\ell \in \A_j}e^{z_{j \ell}}}-\sum_{j \in V} \sum_{\beta \in \A_j} A^{ij}_{1 \beta} \frac{e^{z_{j \beta}}}{\sum_{\ell \in \A_j} e^{z_{j \ell}}}.
\end{equation*}
It follows that for any agent $i \in V$, 
\begin{equation*}
\sum_{\alpha \in \A_i}\frac{d }{dz_{i\alpha}}F_{i\alpha}(z)
= \sum_{\alpha \in \A_i} \sum_{j \in V}\sum_{\beta \in \A_j}A^{ij}_{\alpha \beta} \frac{d}{dz_{i\alpha}}\frac{ e^{z_{j\beta}}}{\sum_{\ell\in \A_j} e^{z_{j\ell}}}- \sum_{\alpha \in \A_i} \sum_{j \in V}\sum_{\beta \in \A_j} A^{ij}_{1 \beta} \frac{d}{dz_{i\alpha}}\frac{ e^{z_{j\beta}}}{ \sum_{\ell \in \A_j} e^{z_{j\ell}}}.    
\end{equation*}
Moreover, observe that for $i \neq j$, 
\begin{equation*}
\frac{d}{dz_{i\alpha}} \frac{ 
e^{z_{j\beta}}}{\sum_{\ell \in \A_j} e^{z_{j \ell}}} =0.
\end{equation*}
Consequently, we get that
\begin{align*}
\sum_{\alpha \in \A_i}\frac{d }{dz_{i\alpha}}F_{i\alpha}(z)
= \sum_{\alpha \in \A_i} \sum_{\beta \in \A_i}A^{ii}_{\alpha \beta} \frac{d}{dz_{i\alpha}}\frac{ e^{z_{i\beta}}}{\sum_{\ell \in \A_i} e^{z_{i\ell}}}- \sum_{\alpha \in \A_i} \sum_{\beta \in \A_i} A^{ii}_{1 \beta} \frac{d}{dz_{i\alpha}}\frac{ e^{z_{i\beta}}}{ \sum_{\ell \in \A_i} e^{z_{i\ell}}}.
\end{align*}
We now separate each sum over $\beta \in \A_i$ into a pair of sums over $\beta\neq \alpha$ and $\beta=\alpha$ for $\alpha \in \A_i$ and any $i \in V$ to get that
\begin{align}
\sum_{\alpha \in \A_i}\frac{d }{dz_{i\alpha}}F_{i\alpha}(z)
&= \sum_{\alpha \in \A_i} \sum_{\beta \neq \alpha}A^{ii}_{\alpha \beta} \frac{d}{dz_{i\alpha}}\frac{ e^{z_{i\beta}}}{\sum_{\ell \in \A_i} e^{z_{i\ell}}}- \sum_{\alpha \in \A_i} A^{ii}_{\alpha \alpha} \frac{d}{dz_{i\alpha}}\frac{ e^{z_{i\alpha}}}{\sum_{\ell \in \A_i} e^{z_{i\ell}}}\notag \\& \qquad- \sum_{\alpha \in \A_i} \sum_{\beta \neq \alpha} A^{ii}_{1 \beta} \frac{d}{dz_{i\alpha}}\frac{ e^{z_{i\beta}}}{ \sum_{\ell \in \A_i} e^{z_{i\ell}}}- \sum_{\alpha \in \A_i} A^{ii}_{1 \alpha} \frac{d}{dz_{i\alpha}}\frac{ e^{z_{i\alpha}}}{ \sum_{\ell \in \A_i} e^{z_{i\ell}}}.
\label{eq:separated}
\end{align}
Recall that the self-loops are antisymmetric, which means that $A^{ii}_{\alpha \alpha}=0$ for any $\alpha \in \A_i$ and $i \in V$. From this property of the game class, 
\begin{equation*}
\sum_{\alpha \in \A_i} A^{ii}_{\alpha \alpha} \frac{d}{dz_{i\alpha}}\frac{ e^{z_{i\alpha}}}{\sum_{\ell \in \A_i} e^{z_{i\ell}}} = 0.
\end{equation*}
Accordingly, an equivalent form of~\eqref{eq:separated} is the expression
\begin{equation}
\sum_{\alpha \in \A_i}\frac{d }{dz_{i\alpha}}F_{i\alpha}(z)
= \sum_{\alpha \in \A_i} \sum_{\beta \neq \alpha}A^{ii}_{\alpha \beta} \frac{d}{dz_{i\alpha}}\frac{ e^{z_{i\beta}}}{\sum_{\ell \in \A_i} e^{z_{i\ell}}}- \sum_{\alpha \in \A_i} \sum_{\beta \neq \alpha} A^{ii}_{1 \beta} \frac{d}{dz_{i\alpha}}\frac{ e^{z_{i\beta}}}{ \sum_{\ell \in \A_i} e^{z_{i\ell}}} - \sum_{\alpha \in \A_i} A^{ii}_{1 \alpha} \frac{d}{dz_{i\alpha}}\frac{ e^{z_{i\alpha}}}{ \sum_{\ell \in \A_i} e^{z_{i\ell}}}.
\label{eq:separated_dropped}
\end{equation}
The derivatives in~\eqref{eq:separated_dropped} are given by
\begin{equation*}
\frac{d}{dz_{i\alpha}}\frac{ e^{z_{i\beta}}}{\sum_{\ell \in \A_i} e^{z_{i\ell}}} = 
\begin{cases}
\frac{\sum_{\ell \in \A_i}e^{z_{i\alpha} + z_{i\ell}} - e^{z_{i\alpha}+z_{i\alpha}}}{(\sum_{\ell \in \A_i} e^{z_{i\ell}})^2} , &  \alpha = \beta \\
-e^{z_{i\beta} + z_{i\alpha}}/(\sum_{\ell \in \A_i} e^{z_{i\ell}})^2, & \alpha \neq \beta.
\end{cases}
\end{equation*}
Evaluating the derivatives in~\eqref{eq:separated_dropped}, we get that
\begin{equation}
\begin{split}
\sum_{\alpha \in \A_i}\frac{d }{dz_{i\alpha}}F_{i\alpha}(z)&= -\sum_{\alpha \in \A_i} \sum_{\beta
\neq \alpha}A^{ii}_{\alpha\beta} 
\frac{e^{z_{i\beta} + z_{i\alpha}}}{(\sum_{\ell \in \A_i} e^{z_{i\ell}})^2}\\
&\quad + \sum_{\alpha \in \A_i} \sum_{\beta \neq \alpha} A^{ii}_{1\beta }
\frac{e^{z_{i\beta} + z_{i\alpha}}}{(\sum_{\ell \in \A_i} e^{z_{i\ell}})^2} -\sum_{\alpha \in \A_i}A_{1\alpha}^{ii} \frac{\sum_{\beta\in \A_i}e^{z_{i\beta}+z_{i\alpha}} - e^{z_{i\alpha}+z_{i\alpha}}}{(\sum_{\ell \in \A_{i}}e^{z_{i\ell}})^2}.
\end{split}
\label{eq:precombine}
\end{equation}

Moreover, from a series of algebraic manipulations, we find that
\begin{align*}
&\sum_{\alpha \in \A_i} \sum_{\beta \neq \alpha} A^{ii}_{1\beta }
\frac{e^{z_{i\beta} + z_{i\alpha}}}{(\sum_{\ell \in \A_i} e^{z_{i\ell}})^2}-\sum_{\alpha \in \A_i} \A_{1\alpha}^{ii} \frac{\sum_{\beta\in \A_i}e^{z_{i\beta}+z_{i\alpha}} - e^{z_{i\alpha}+z_{i\alpha}}}{(\sum_{\ell \in \A_{i}}e^{z_{i\ell}})^2} \\
&= \sum_{\alpha \in \A_i} \sum_{\beta \neq \alpha} A^{ii}_{1\beta }
\frac{e^{z_{i\beta} + z_{i\alpha}}}{(\sum_{\ell \in \A_i} e^{z_{i\ell}})^2}+ \sum_{\alpha \in \A_i}A_{1\alpha}^{ii} \frac{e^{z_{i\alpha}+z_{i\alpha}}}{(\sum_{\ell \in \A_{i}}e^{z_{i\ell}})^2} - \sum_{\alpha \in \A_i}\sum_{\beta\in \A_i}A_{1\alpha}^{ii} \frac{e^{z_{i\beta}+z_{i\alpha}}}{(\sum_{\ell \in \A_i}e^{z_{i\ell}})^2} \\
&= \sum_{\alpha \in \A_i} \sum_{\beta \in  \A_i} A^{ii}_{1\beta }
\frac{e^{z_{i\beta} + z_{i\alpha}}}{(\sum_{\ell \in \A_i} e^{z_{i\ell}})^2}- \sum_{\alpha \in \A_i}\sum_{\beta\in \A_i}A_{1\alpha}^{ii} \frac{e^{z_{i\beta}+z_{i\alpha}}}{(\sum_{\ell \in \A_i}e^{z_{i\ell}})^2} \\
&= 0. 
\end{align*}
It follows that~\eqref{eq:precombine} is equivalent to
\begin{equation}
\begin{split}
\sum_{\alpha \in \A_i}\frac{d }{dz_{i\alpha}}F_{i\alpha}(z) &= -\sum_{\alpha \in \A_i} \sum_{\beta
\neq \alpha}A^{ii}_{\alpha\beta} 
\frac{e^{z_{i\beta} + z_{i\alpha}}}{(\sum_{\ell \in \A_i} e^{z_{i\ell}})^2}.
\end{split}
\label{eq:prefinal}
\end{equation}
Finally, we reorganize the sums in~\eqref{eq:prefinal} over pairs $(\alpha,\beta)$ such that $\beta\neq \alpha$ and invoke the fact that each matrix $A^{ii}$ is antisymmetric (meaning that $(A^{ii})^\top = - A^{ii}$) to conclude
\begin{equation}
\sum_{\alpha \in \A_i}\frac{d }{dz_{i\alpha}}F_{i\alpha}(z) = -\sum_{\alpha \in \A_i} \sum_{\beta
\neq \alpha}A^{ii}_{\alpha\beta} 
\frac{e^{z_{i\beta} + z_{i\alpha}}}{(\sum_{\ell \in \A_i} e^{z_{i\ell}})^2} 
= \sum_{(\alpha,\beta):~\beta \neq \alpha}
(- A_{\alpha \beta}^{ii} -A_{\beta \alpha}^{ii}) \frac{e^{z_{i\alpha}} + e^{z_{i\beta}}}{(\sum_{\ell \in \A_i}e^{z_{i\ell}})^2} \notag 
=0.
 \label{eq:final_zero}
\end{equation}
So, by summing \eqref{eq:final_zero} over $i\in V$, we obtain \begin{equation*}
\sum_{i \in V}\sum_{\alpha \in \A_i}\frac{d }{dz_{i\alpha}}F_{i\alpha}(z)=0.
\end{equation*}

\subsubsection{Proof of Lemma~\ref{t:invariant}}
\label{appsecs:invariant_proof}
Consider an $N$-player rescaled zero-sum polymatrix game with an interior Nash equilibrium such that for positive coefficients $\{\sconst\}_{i\in V}$, $\sum_{i\in V} \sconst_i  u_i(x) =0 $ for any $x \in \X$. We need to show the function
\begin{equation}
\Phi(t) = \sum_{i \in V} \sum_{\alpha \in \A_i}\sconst_i  x_{i\alpha}^\ast \ln x_{i\alpha}
\label{eq:phi}
\end{equation}
is time invariant for any trajectory generated by the replicator dynamics, meaning that $\Phi(t) =\Phi(0)$ for all $t\geq 0$.

In order to prove $\Phi(t)$ as given in~\eqref{eq:phi} is time-invariant, we show that the time derivative of the function is equal to zero. 
To begin, recall the form of the replicator dynamics from~\eqref{eq:replicator} given by
\begin{equation}
\dot{x}_{i \alpha} = x_{i \alpha} 
(u_{i\alpha}(x) - u_i(x)),\ \ \forall \alpha\in \A_i.
\label{eq:replicator_repeat}
\end{equation}
We simplify the time derivative of $\Phi(t)$ using the structure of the dynamics given in~\eqref{eq:replicator_repeat} as follows:
\begin{align}
\frac{d \Phi(t)}{dt} &= \sum_{i \in V}\sconst_i\sum_{\alpha \in \A_i} x_{i\alpha}^{\ast}\frac{d\ln x_{i\alpha}}{dt} \notag \\
&= \sum_{i \in V}\sconst_i\sum_{\alpha \in \A_i} x_{i\alpha}^{\ast}\frac{\dot{x}_{i\alpha}}{x_{i\alpha}} \notag  \\
&= \sum_{i \in V}\sconst_i\sum_{\alpha \in \A_i} x_{i\alpha}^{\ast}(u_{i\alpha}(x) - u_i(x))\notag \\
&= \sum_{i \in V}\sconst_i\sum_{\alpha \in \A_i} x_{i\alpha}^{\ast} u_{i\alpha}(x) - \sum_{i\in V}\eta_i\sum_{\alpha \in \A_i}x_{i\alpha}^{\ast}u_i(x). \label{eq:invariant_simplify}
\end{align}
Let $E' = \{(i, j): i\neq j, (i, j)\in E\}$ denote the edge set of the polymatrix game excluding self-loops. In the remainder of the proof, denote by $e_{\alpha}$ a one-hot vector of appropriate dimension containing all zeros, except for a one in the $\alpha$--th entry. Furthermore, recall $u_{i\alpha}(x)$ denotes the utility of player $i\in V$ for playing the pure strategy $\alpha \in \A_i$, which can be represented by $x_i=e_{\alpha}$, when the rest of the agents play $x_{-i}$.
Then, from the fact that $A^{ii}_{\alpha\alpha} = 0$ for all $\alpha \in \A_i$ and $i \in V$, we obtain
\begin{align}
\sum_{i \in V}\sconst_i\sum_{\alpha \in \A_i} x_{i\alpha}^{\ast} u_{i\alpha}(x) &= \sum_{i \in V}\sconst_i\sum_{\alpha \in \A_i} x_{i\alpha}^{\ast}  \sum_{j: (i, j)\in E}  (A^{ij}x_j)_{\alpha}\notag  \\
&= \sum_{i\in V}\eta_i\sum_{\alpha \in \A_i} x_{i\alpha}^\ast (A^{ii}_{\alpha}e_\alpha)_\alpha+\sum_{i \in V}\sconst_i\sum_{\alpha \in \A_i} x_{i\alpha}^{\ast}  \sum_{j: (i, j)\in E'}  (A^{ij}x_j)_{\alpha} \notag  \\
&= \sum_{i\in V}\eta_i\sum_{\alpha \in \A_i} x_{i\alpha}^\ast A^{ii}_{\alpha\alpha}+\sum_{i \in V}\sconst_i\sum_{\alpha \in \A_i} x_{i\alpha}^{\ast}  \sum_{j: (i, j)\in E'}  (A^{ij}x_j)_{\alpha} \notag  \\
&= \sum_{i \in V}\sconst_i (x_{i}^{\ast})^{\top}  \sum_{j: (i, j)\in E'}  A^{ij}x_j.
\label{eq:drop_loop}
\end{align}

Moreover, since $\sum_{\alpha \in \A_i}x_{i\alpha}^{\ast}=1$ for each $i \in V$ and $\sum_{i\in V}\eta_iu_i(x)=0$ for any strategy profile $x \in \X$ from the game being rescaled zero-sum, we get that
\begin{equation}
\sum_{i\in V}\eta_i\sum_{\alpha \in \A_i}x_{i\alpha}^{\ast}u_i(x)= \sum_{i\in V}\eta_iu_i(x)\sum_{\alpha \in \A_i}x_{i\alpha}^{\ast}= \sum_{i\in V}\eta_iu_i(x)=0.
\label{eq:zero_trick}
\end{equation}
Combining~\eqref{eq:invariant_simplify},~\eqref{eq:drop_loop}, and~\eqref{eq:zero_trick}, we have
\begin{equation}
\frac{d \Phi(t)}{dt}  = \sum_{i \in V} \sconst_i (x_i^*)^\top \sum_{j:(i,j) \in E'}A^{ij}x_j.
\label{eq:pre_theorem1}
\end{equation}
For the interior Nash equilibrium $x^{\ast}$ under consideration,
\begin{align}
\sum_{i\in V}\eta_i u_i(x^\ast) &= \sum_{i\in V}\eta_i (x_i^\ast)^\top \sum_{j:(i, j)\in E}A^{ij}x_j^\ast \notag \\
&= \sum_{i\in V}\eta_i (x_i^\ast)^\top \sum_{j:(i, j)\in E'}A^{ij}x_j^\ast \label{eq:rescale_zero} \\
&=0. \label{eq:rescale_zero2}
\end{align}
Note that~\eqref{eq:rescale_zero} holds from the fact that $(x_i^{\ast})^{\top}A^{ii} x_i^{\ast}=0$ for all $x_i^{\ast}\in \X_i$ and $i \in V$ since the self-loops are antisymmetric and~\eqref{eq:rescale_zero2} as a result of the polymatrix game being rescaled zero-sum.
We continue by subtracting~\eqref{eq:rescale_zero} from~\eqref{eq:pre_theorem1} since it is equal to zero and obtain
\begin{align}
\frac{d \Phi(t)}{dt} = \sum_{i \in V} \sconst_i (x_i^*)^\top \sum_{j:(i,j) \in E'}A^{ij}x_j-\sum_{i\in V}\eta_i (x_i^\ast)^\top \sum_{j:(i, j)\in E'}A^{ij}x_j^\ast=\sum_{i\in V}\eta_i\sum_{j:(i, j)\in E'}(x_i^\ast)^\top A^{ij}(x_j-x_j^\ast).
\label{eq:pre_theorem2}
\end{align}
We now prove that~\eqref{eq:pre_theorem2} is equal to zero. To do so, we rely on the results of \citet[Section 4]{cai2011minmax}, who show that any zero-sum polymatrix game without self-loops can be transformed to a payoff equivalent, pairwise constant-sum game. 
Indeed, the results of~\citet{cai2011minmax} apply to rescaled zero-sum polymatrix games since for any strategy profile $x \in \X$, 
\begin{align*}
\sum_{i \in V}\eta_i u_i(x) = \sum_{i\in V}x_i^\top \sum_{j:(i, j)\in E} \eta_iA^{ij}  x_j= \sum_{i\in V}x_i^\top \sum_{j:(i, j)\in E'} \eta_iA^{ij}  x_j  = 0.
\end{align*}
This means that for each edge $(i,j) \in E'$ there exists a matrix $B^{ij}$ such that the following properties hold (see Lemma~3.1, 3.2, and 3.4, respectively in \cite{cai2011minmax}):
\begin{description}
    \item[Property 1.] $\eta_i A^{ij}_{\alpha \beta} - \eta_i A^{ij}_{\alpha \gamma} = B^{ij}_{\alpha \beta} - B^{ij}_{\alpha \gamma}$ for any $\alpha \in \A_i$ and $\beta,\gamma \in \A_j$.

    \item[Property 2.] $B^{ij} + (B^{ji})^\top = c_{ij} \cdot \onev_{n_i \times n_j}$, where $\onev_{n_i \times n_j}$ is an $n_i\times n_j$ matrix of all ones. 
    
    \item[Property 3.] In every joint pure strategy profile, every player $i \in V$ has the same utility in the game defined by the payoff matrices $\{\eta_i A^{ij}\}_{(i,j)\in E'}$ as in the game defined by the payoff matrices $\{B^{ij}\}_{(i,j)\in E'}$ . 
\end{description}
\smallskip
\smallskip
\smallskip
\smallskip
\noindent Fixing a strategy $\gamma \in \A_j$, we can express the summand of~\eqref{eq:pre_theorem2} using Property 1~as follows:
\begin{align}
(x_i^\ast)^\top  \eta_i A^{ij}  (x_j-x_j^\ast) &= 
\sum_{\alpha \in \A_i}\sum_{\beta \in \A_j}
x_{i\alpha}^\ast \eta_i A^{ij}_{\alpha \beta}(x_{j\beta}-x_{j\beta}^\ast) \notag \\
&= 
\sum_{\alpha \in \A_i}\sum_{\beta \in \A_j}
x_{i\alpha}^\ast \big( B^{ij}_{\alpha \beta} - B^{ij}_{\alpha \gamma} +  \eta_i A^{ij}_{\alpha \gamma} \big) (x_{j_\beta}-x_{j_\beta}^\ast) \notag \\
&= (x_i^\ast)^\top  B^{ij}  (x_j-x_j^\ast) + \sum_{\alpha \in \A_i} x_{i\alpha}\left(\eta_i A^{ij}_{\alpha \gamma}-B^{ij}_{\alpha \gamma} \right) \sum_{\beta \in \A_j}(x_{j\beta}-x_{j\beta}^\ast). 
\label{eq:property1_last}
\end{align}
Moreover, observe that since both $x_j$ and $x_j^{\ast}$ are on the simplex, $\sum_{\beta \in \A_j}(x_{j\beta}-x_{j\beta}^\ast)=0$, and consequently
\begin{equation}
\sum_{\alpha \in \A_i} x_{i\alpha}\left(- B^{ij}_{\alpha \gamma} +  \eta_i A^{ij}_{\alpha \gamma} \right) \sum_{\beta \in \A_j}(x_{j\beta}-x_{j\beta}^\ast) = 0.
\label{eq:property1_last2}
\end{equation}
Then, relating~\eqref{eq:property1_last} and~\eqref{eq:property1_last2}, we obtain
\begin{equation*}
(x_i^\ast)^\top  \eta_i A^{ij}  (x_j-x_j^\ast) =(x_i^\ast)^\top  B^{ij}  (x_j-x_j^\ast).
\end{equation*}
As a result, \eqref{eq:pre_theorem2} is equivalently
\begin{equation}
\frac{d \Phi(t)}{dt}
=\sum_{i\in V} \sum_{j:(i, j)\in E'}(x_i^\ast)^\top  \eta_i A^{ij}  (x_j-x_j^\ast).
\end{equation}
Then, swapping the sum indexing and taking the transpose of the quadratic form $(x_i^\ast)^\top B^{ij}(x_j-x_j^\ast)$, 
\begin{align}
\frac{d \Phi(t)}{dt} &= \sum_{j \in V}\sum_{i:(j, i)\in E'}(x_i^\ast)^\top  B^{ij}  (x_j - x_j^\ast)\notag\\
&= \sum_{i \in V}\sum_{i:(j, i)\in E'}(x_j - x_j^\ast)^\top  (B^{ij})^\top x_i^\ast. \notag 
\end{align}
We now invoke Property 2 to replace $(B^{ij})^\top$ with $c^{ji}\onev_{n_j\times n_i}-B^{ji}$ in the previous equation, which results in 
\begin{equation}
\frac{d \Phi(t)}{dt}
= \sum_{j \in V}\sum_{i:(j, i)\in E'}(x_j - x_j^\ast)^\top  \left (c^{ji}  \onev_{n_j \times n_i}  - B^{ji} \right )  x_i^\ast
\label{eq:sumindex}
\end{equation}
For any $x_j \in \X_j$, $x_j^{\ast} \in \X_j$ and $x_i^{\ast} \in \X_i$, we have
\begin{equation*}
c^{ji} (x_j - x_j^\ast)^\top \onev_{n_j \times n_i}  x_i^\ast = c^{ji} (x_j - x_j^\ast)^\top \onev_{n_j} = c^{ji}-c^{ji}  = 0,
\end{equation*}
since $\X_j=\Delta^{n_j}$ and $\X_i=\Delta^{n_i}$ so that $\sum_{\alpha\in \A_j}x_{j\alpha} = \sum_{\alpha\in \A_j}x_{j\alpha}^{\ast} = \sum_{\alpha\in \A_i}x_{i\alpha}^{\ast} =1$.
Accordingly, we simplify~\eqref{eq:sumindex} and get 
\begin{equation}
\frac{d \Phi(t)}{dt} = -\sum_{j \in V}\sum_{i:(j, i)\in E'}(x_j - x_j^\ast)^\top  B^{ji}  x_i^\ast.
\label{eq:sumindex2}
\end{equation}

Following a similar argument as above, we analyze the summand in \eqref{eq:sumindex2} for some $j \in V$. Using 
 Property 1 and fixing any strategy $\gamma_i\in \A_i$ for each $i\in V\setminus \{j\}$, we have that
\begin{align}
 \sum_{i:(j,i)\in E'}(x_j-x_j^\ast)^\top B^{ji}x_i^\ast&=\sum_{i:(j,i)\in E'}\sum_{\alpha\in \A_j}\sum_{\beta\in \A_i}  z_{j\alpha} B^{ji}_{\alpha\beta}x_{i\beta}^\ast \notag \\
&=\sum_{i:(j,i)\in E'}\sum_{\alpha\in \A_j}\sum_{\beta\in \A_i} z_{j\alpha} \big(\eta_jA^{ji}_{\alpha\beta}-\eta_jA^{ji}_{\alpha\gamma_i}+B^{ji}_{\alpha\gamma_i}\big)x_{i\beta}^\ast \notag \\
&=\sum_{i:(j,i)\in E'}\eta_j(x_j-x_j^\ast)^\top A^{ji}x_i^\ast  +\sum_{i:(j,i)\in E'}\sum_{\alpha\in \A_j}\sum_{\beta\in\A_i}z_{j\alpha} (-\eta_jA^{ji}_{\alpha\gamma_i}+B_{\alpha\gamma_i}^{ji})x_{i\beta}^\ast.
\label{eq:pre_sep}
\end{align}
where $z_{j\alpha}:=x_{j\alpha}-x_{j\alpha}^\ast$.
We now examine the last term in the equation overhead, and use the fact $\sum_{\beta\in\A_i}x_{i\beta}^\ast=1$ since $x_i^{\ast}\in \X_i^{\ast}=\Delta^{n_i-1}$ to get that
\begin{align}
\sum_{i:(j,i)\in E'}\sum_{\alpha\in \A_j}\sum_{\beta\in\A_i}(x_{j\alpha}-x_{j\alpha}^\ast) (-\eta_jA^{ji}_{\alpha\gamma_i}+B_{\alpha\gamma_i}^{ji})x_{i\beta}^\ast 
&= \sum_{i:(j,i)\in E'}\sum_{\alpha\in \A_j}(x_{j\alpha}-x_{j\alpha}^\ast) (-\eta_jA^{ji}_{\alpha\gamma_i}+B_{\alpha\gamma_i}^{ji})\sum_{\beta\in\A_i}x_{i\beta}^\ast \notag \\
&= \sum_{\alpha\in \A_j}(x_{j\alpha}-x_{j\alpha}^\ast) \sum_{i:(j,i)\in E'}(-\eta_jA^{ji}_{\alpha\gamma_i}+B_{\alpha\gamma_i}^{ji}) 
\label{eq:pre_zero}
\end{align}
For each $\alpha \in \A_j$, the terms $\sum_{i:(j,i)\in E'}\eta_jA^{ji}_{\alpha\gamma_i}$ and $\sum_{i:(j,i)\in E'}B^{ji}_{\alpha\gamma_i}$ give the utility of player $j \in V$ in the games defined by $\{\eta_j A^{ji}\}_{(j,i)\in E'}$ and $\{B^{ji}\}_{(j,i)\in E'}$ under a pure strategy profile such that agent $j$ plays $\alpha$ and each other agent $i \in V\setminus \{j\}$ plays some $\gamma_i \in \A_i$. From Property 3, we conclude for each $\alpha \in \A_j$ that
\begin{equation}
 \sum_{i:(j,i)\in E'}(-\eta_jA^{ji}_{\alpha\gamma_i}+B_{\alpha\gamma_i}^{ji}) =0.
\label{eq:zero}
\end{equation}
Relating~\eqref{eq:zero} back to~\eqref{eq:pre_zero} and then~\eqref{eq:pre_sep}, for each $j \in V$, we obtain
\begin{equation}
\sum_{i:(j,i)\in E'}(x_j-x_j^\ast)^\top B^{ji}x_i^\ast = \sum_{i:(j,i)\in E'}\eta_j(x_j-x_j^\ast)^\top A^{ji}x_i^\ast.
\label{eq:finaleq}
\end{equation}

Finally, combining~\eqref{eq:finaleq} and~\eqref{eq:sumindex2}, we have 
\begin{equation*}
\frac{d \Phi(t)}{dt}
=-\sum_{j \in V}\sum_{i:(j, i)\in E}\eta_j(x_j-x_j^\ast)^\top A^{ji}x_i^\ast=0
\end{equation*}
where the final equality holds since $x^\ast$ is an interior Nash equilibrium, which means $u_{j\alpha}(x^\ast)=u_j(x^\ast)$ for all strategies $\alpha\in \A_j$ and any linear combination thereof.
Consequently, we conclude that $\Phi(t) =\Phi(0)$ for all $t\geq 0$.

\paragraph{Orbits remain bounded away from the boundary.} To complete the proof, we use the constant of motion to show that the orbits of the replicator dynamics for rescaled zero sum polymatrix games remain bounded away from the boundary. Indeed, let $x$ be an interior point which is not an equilibrium. That is, each $x_i\in \mathrm{int}(\Delta^{n-1})$. Let $\gamma(x)$ be the forward orbit of $x$ i.e., \[\gamma(x)=\{\phi^t(x):\ t\geq 0\}\] Then, Lemma~\ref{t:invariant} implies that for any $y\in \gamma(x)$, \[-c=\Phi(t)=\sum_{i\in V}\sum_{\alpha\in\A_i}\eta_ix_{i\alpha}^\ast\ln x_{i\alpha}=\sum_{i\in V}\sum_{\alpha\in\A_i}\eta_ix_{i\alpha}^\ast\ln y_{i\alpha}<0\] since $\Phi(t)$ is a constant of motion. 
For any $i$ and any $y\in \gamma(x)$, $\sum_{\alpha\in \A_i}\eta_ix_{i\alpha}^\ast\ln y_{i\alpha}\in [-c,0]$, since $\sum_{\alpha\in \A_i}\eta_ix_{i\alpha}^\ast\ln y_{i\alpha}\leq 0$ for any $y$ and any player $i$. Hence, for any $\beta\in \A_i$,
\[-c\leq -c-\sum_{\alpha\neq \beta}x_{i\alpha}^\ast\ln y_{i\alpha}<x_{i\beta}^\ast\ln y_{i\beta}\]
since $\sum_{\alpha\neq \beta}x_{i\alpha}^\ast\ln y_{i\alpha}\leq 0$. This implies that
\[y_{i\beta}\geq \exp\big(-c/x_{i\beta}^\ast\big).\]
Let \[\varepsilon=\min_{i\in V,\beta\in \A_i}\exp\big(-c/x_{i\beta}^\ast\big)\]
so that for any $y\in \gamma(x)$ and any player $i\in V$ and strategy $\beta\in \A_i$, $y_{i\beta}\geq \varepsilon>0$ which, in turn, implies that $\gamma(x)$ is bounded away from the boundary.

\paragraph{KL-Divergence Constant of Motion.}
The constant of motion from Lemma~\ref{t:invariant} is equivalently given as 
\[\Phi(x)=-\sum_{i\in V}\eta_i(\mathrm{KL}(x^*_i||x_i)-\mathcal{I}(x^*_i))\]
where $\mathcal{I}(\cdot)$ denotes the entropy and $\mathrm{KL}(\cdot||\cdot)$ denotes the Kullback-Leibler (KL) divergence.  Since $\sum_{i\in V}\eta_i\mathcal{I}(x^*_i))$ is a constant, it does not change the time-invariant property, and hence the weighted sum of KL divergences is itself a constant of motion. 
\begin{corollary}
Under the assumptions of Lemma~\ref{t:invariant}, 
$\Psi(t)=\sum_{i\in V}\eta_i(\mathrm{KL}(x^*_i||x_i)$ is also a constant of motion.
\label{cor:newconstant}
\end{corollary}

\subsection{Proofs of Time-Average Convergence, Equilibrium Computation, \& Bounded Regret from Section~\ref{sec:timeaverage}}
\label{appsecs:regret}
We now provide the proofs of Theorem~\ref{t:avg1}, Theorem~\ref{thm:lp}, and Proposition~\ref{thm:regretbdd} from Section~\ref{sec:timeaverage}. Appendix~\ref{appsecs:average_proof} contains the proof of Theorem~\ref{t:avg1}, which shows the time-average equilibria and utility convergence of the replicator dynamics in rescaled zero-sum polymatrix games. The proof of Theorem~\ref{thm:lp}, which provides a linear program to compute Nash equilibrium in rescaled zero-sum polymatrix games is in Appendix~\ref{appsecs:lp_proof}. Finally, the proof of Proposition~\ref{thm:regretbdd}, which shows that the replicator dynamics achieve bounded regret, is given in Appendix~\ref{appsecs:regret_proof}. 
\subsubsection{Proof of Theorem~\ref{t:avg1}}
\label{appsecs:average_proof}

Let $x^\ast$ denote the unique Nash equilibrium of the game. Recall that the trajectory $x(t)$ remains on the interior of the simplex for all $t\geq 0$ as a result of Lemma~\ref{t:invariant}. Integrating the replicator dynamics from~\eqref{eq:replicator} given by
\[\dot{x}_{i \alpha}(t) = x_{i \alpha}(t) 
(u_{i\alpha}(x(t)) - u_i(x(t))\]
for each agent $i \in V$ and each strategy $\alpha \in \A_i$, we obtain
\begin{equation*}
   \frac{1}{T}\int_{x(0)}^{x(T)} \frac{1}{x_{i\alpha}(\tau)} d x(\tau) =  \frac{1}{T}\int_{0}^T \left( u_{i\alpha}(x(\tau)) - u_{i}(x(\tau))
\right) d \tau .
\end{equation*}
Furthermore, 
\begin{equation*}
      \frac{1}{T}\int_{x(0)}^{x(T)} \frac{1}{x_{i\alpha}(\tau)} d x(\tau)=\frac{1}{T}\left( \log x_{i\alpha}(T) - \log x_{i\alpha}(0) \right) 
\end{equation*}
so that
\begin{equation}
 \frac{1}{T}\int_{0}^T \left( u_{i\alpha}(x(\tau)) - u_{i}(x(\tau))
\right) d \tau =\frac{1}{T}\left( \log x_{i\alpha}(T) - \log x_{i\alpha}(0) \right).  \label{eq:prelimit}
\end{equation}

 Define
\[z_{i\alpha}(T)=\frac{1}{T}\int_0^Tx_{i\alpha}(\tau)\ d\tau.\]
Clearly $z_{i\alpha}(T)$ is bounded for all $T$ since $x_{i\alpha}(T)$ remains bounded. Moreover, the bounds on $z_{i\alpha}(T)$ are the same as those on $x_{i\alpha}(T)$. Consider any sequence $T_k$ converging to infinity. The Bolzano–Weierstrass theorem guarantees that the bounded sequence $z_{i\alpha}(T_k)$ admits a convergent subsequence $z_{i\alpha}(T_{k_\ell})$
such that $z_{i\alpha}(T_{k_\ell})$ converges towards some limit which we denote by $\bar{x}_{i\alpha}$. Since we can repeat this argument for all $i\in V$ and all $\alpha\in \A_i$, let $\bar{x}_i=(\bar{x}_1, \ldots, \bar{x}_{n_i})$ for each $i\in V$. 

The sequences $\log( x_{i\alpha}(T_{k}))-\log(x_{i\alpha}(0))$ are also bounded.  Passing to the limit in \eqref{eq:prelimit} and using the fact that $x_{i\alpha}(t)$ remains bounded away from zero for all $t\geq 0$, for each $i \in V$, we have that
\begin{equation}\lim_{T\to \infty}\frac{1}{T}\int_{0}^T \left( u_{i\alpha}(x(\tau)) - u_{i}(x(\tau))
\right) d \tau=0, \ \ \forall \alpha\in \A_i.\label{eq:passingtolimit}
\end{equation}
Rearranging \eqref{eq:passingtolimit}, for each $i \in V$, we have that
\begin{align}
  \lim_{T\to \infty}\frac{1}{T}\int_{0}^T   u_{i}(x(\tau))d\tau &=\lim_{T\to \infty}\frac{1}{T}\int_{0}^T u_{i\alpha}(x(\tau))d\tau, \notag \ \ \forall \alpha\in \A_i\\
  &=\lim_{T\to \infty}\frac{1}{T}\int_{0}^T\sum_{j:(i,j)\in E}(A^{ij}x_j(\tau))_\alpha d\tau, \notag \ \ \forall \alpha\in \A_i\\
  &=\sum_{j:(i,j)\in E}(A^{ij}\bar{x}_j)_\alpha, \ \ \forall \alpha \in \A_i,\label{eq:limsumtwo}
\end{align}
where the last equality follows from linearity of the integral, finiteness of the sum, and the well-defined limit. Hence, weighting by $\bar{x}_{i\alpha}$ and summing across $\alpha\in \A_i$, we have that 
\begin{align*}
    u_i(\bar{x})&=\sum_{\alpha\in\A_i}\bar{x}_{i\alpha} \sum_{j:(i,j)\in E}(A^{ij}\bar{x}_j)_\alpha\\
    &=\sum_{\alpha\in\A_i}\bar{x}_{i\alpha}\left( \lim_{T\to \infty}\frac{1}{T}\int_{0}^T   u_{i}(x(\tau))d\tau\right)\\
    &=\lim_{T\to \infty}\frac{1}{T}\int_{0}^T   u_{i}(x(\tau))d\tau
\end{align*}
where the last equality holds since $\sum_{\alpha\in \A_i}\bar{x}_{i\alpha}=1$.
In turn, the above implies that
\[u_i(\bar{x})=\sum_{j:(i,j)\in E}(A^{ij}\bar{x}_j)_\alpha=u_{i\alpha}(\bar{x}), \ \forall \alpha\in \A_i\]
so that $\bar{x}=(\bar{x}_1, \ldots, \bar{x}_N)$ is a Nash Equilibrium.  Since there exists a unique Nash equilibrium by assumption, we have that $\bar{x} = x^\ast$ which proves $(i)$. Combining this fact with \eqref{eq:limsumtwo}, 
we have that
\[\lim_{T\rightarrow \infty} \frac{1}{T}\int_{0}^{T} u_{i}(x(\tau))d\tau = u_{i\alpha}(x^\ast) = u_{i}(x^\ast)\] 
which proves $(ii)$.

\subsubsection{Proof of Theorem~\ref{thm:lp}}
\label{appsecs:lp_proof}
Let $\mathrm{OPT}$ denote the optimal value of the linear program
\begin{equation}
\begin{split}
\min_{x \in \X} \quad&   \sum_{i=1}^{n} \sconst_iv_i  \\
\text{s.t} \quad  &v_i \geq u_{i\alpha}(x),\ \forall \ i \in V, \  \forall\  \alpha \in \A_i.
\end{split}
\label{eq:lpproof}
\end{equation}
We begin by proving that $\mathrm{OPT} \leq 0$. Since a Nash equilibrium always exists~\citep{nash1951non}, there exists a strategy profile $x$ such that $\max_{\alpha\in \A_i}u_{i\alpha}(x)=u_i(x)$. That is,
\begin{equation} \max_{\alpha \in \A_i}\sum_{j:(i,j)\in E} (A^{ij} x_j)_\alpha = \sum_{\alpha \in \A_i} x_{i\alpha}  \sum_{j:(i,j)\in E} (A^{ij} x_j)_\alpha, \ \forall \ i\in V.\label{eq:feasiblelp}
\end{equation}
Let $v_i = \max_{\alpha \in \A_i}\sum_{j:(i,j)\in E} (A^{ij} x_j)_\alpha$ for all $i \in V$. Then, the pair of vectors $(v,x)$ forms a feasible solution for the linear program in~\eqref{eq:lpproof}. As a result, using~\eqref{eq:feasiblelp}, we have that
\begin{equation*}
 \mathrm{OPT} \leq \sum_{i \in V} \sconst_i v_i = 
 \sum_{i \in V}\sconst_i \sum_{\alpha \in \A_i}x_{i\alpha}\sum_{j:(i,j)\in E} (A^{ij} x_j)_\alpha = 0  
\end{equation*}
where the last equality follows by the fact that $\sum_{i=1}^n \sconst_i u_i(x) =0$ since the game is rescaled zero-sum.

Let $(v^\ast,x^\ast)$ denote the optimal solution of the  linear program in~\eqref{eq:lpproof}. We now prove that $x^\ast$ is a Nash equilibrium using the fact that $\mathrm{OPT}\leq 0$ . For the sake of contradiction, assume $x^\ast$ is not a Nash equilibrium, which would mean there exists an agent $i\in V$ and a strategy $\alpha \in \A_i$ satisfying
\begin{align}
\max_{\alpha\in \A_i}u_{i\alpha}(x^\ast)&=\max_{\alpha \in \A_i}\sum_{j:(i,j)\in E} (A^{ij} x^\ast_j)_\alpha > \sum_{\alpha \in \A_i}x^\ast_{i\alpha}  \sum_{j:(i,j)\in E} (A^{ij} x_j^\ast)_\alpha=u_i(x^\ast).
\label{eq:pre_final_lp}
\end{align}
Moreover, since $(v^\ast,x^\ast)$ is the optimal solution of the  linear program in~\eqref{eq:lpproof}, we know that $v_i^{\ast} \geq u_{i\alpha}(x^{\ast})$ for all $i \in V$ and $\alpha \in \A_i$, which then implies $v_i^{\ast} \geq \max_{\alpha \in \A_i}u_{i\alpha}(x^{\ast})$ for all $i \in V$. As a direct result, we obtain the inequality
\begin{equation}
\mathrm{OPT} =\sum_{i \in V} \sconst_i v_i^\ast
\geq \sum_{i \in V} \sconst_i \max_{\alpha \in \A_i}\sum_{j:(i,j)\in E} (A^{ij} x^\ast_j)_\alpha 
\label{eq:pre_final_lp2}
\end{equation}
Finally, combining~\eqref{eq:pre_final_lp} and~\eqref{eq:pre_final_lp2}, we get that
\begin{equation*}
\mathrm{OPT} > \sum_{i\in V} \sconst_i \sum_{\alpha \in \A_i}x^\ast_{i\alpha} \sum_{j:(i,j)\in E} (A^{ij} x^\ast_j)_\alpha=0,
\end{equation*}
where the last equality follows by the fact that $\sum_{i=1}^n \sconst_i u_i(x) =0$ since the game is rescaled zero-sum.
Yet, this leads to contradiction since $\mathrm{OPT} \leq 0$, which means $x^\ast$ must be a Nash equilibrium.

\subsubsection{Proof of Proposition~\ref{thm:regretbdd}}
\label{appsecs:regret_proof}
We begin by presenting preliminaries and notation needed for an intermediate technical result. Denote by $v_i(x)=(u_{i\alpha}(x))_{\alpha\in \A_i}$ the payoff vector for any agent $i \in V$ that includes the utility of each pure strategy $\alpha\in \A_i$ under the joint profile $x = (\alpha, x_{-i})\in \X$. The utility of the player $i \in V$ under the joint strategy profile $x = (x_i, x_{-i})\in \X$ is then given by $u_i(x) = \langle v_i(x), x_i\rangle$. The learning dynamics given by
\begin{equation}
\begin{split}
y_i(t)&=\int_0^t v_i(x(s)) ds \\
x_i(t)&=Q_i(y_i(t))
\end{split}
\label{eq:fotrl}
\end{equation}
characterize the ``Follow the Regularized Leader'' updates for player $i\in V$ at time $t\geq 0$. The so-called choice map $Q_i:\mathbb{R}^{n_i}\rightarrow \X_i$ is defined by
\begin{equation*}
Q_i(y_i)=\arg\max_{x_i\in \X_i}\{\langle y_i,x_i\rangle-h_i(x_i)\}
\end{equation*}
for a strongly convex and continuously differentiable regularizer function $h_i:\X_i\rightarrow \mathbb{R}$. The strong convexity of $h_i$ along with the convexity and compactness of $\X_i$ ensure a unique solution exists for the update $x_i(t)$ so that it is well-defined. The negative entropy regularizer function
\begin{equation*}
h_i(x_i)=\sum_{\alpha\in \A_i}x_{i\alpha}\log(x_{i\alpha})
\end{equation*}
gives rise to the replicator dynamics we study in this work. Furthermore, the convex conjugate of the regularizer function $h_i$ is given by
\begin{equation*}
h^\ast_i(y_i)=\max_{x_i\in \X_i}\{\langle y_i, x_i\rangle-h_i(x_i)\}.
\end{equation*}
A simple corollary of this definition is Fenchel's inquality, which says for every $x_i \in \X_i$ and $y_i\in \mathbb{R}^{n_i}$, 
\begin{equation*}
\langle y_i, x_i\rangle \leq h_i(x_i)+ h_i^{\ast}(y_i).
\end{equation*}
Moreover, by the maximizing argument (see e.g.,~\cite[Ch.~2]{shalev2012online}), $x_i(t)=Q_i(y_i(t))= \nabla h_i^\ast(y_i(t))$. 

We now state and prove an intermediate result, which we then invoke to conclude Proposition~\ref{thm:regretbdd}.
\begin{lemma} Let $h_{\max, i}=\max_{x_i\in \X_i} h_i(x_i)$ and $h_{\min, i}=\min_{x_i\in \X_i} h_i(x_i)$.
If player $i \in V$ follows the replicator dynamics from~\eqref{eq:replicator}, then independent of the rest of the players in the game,
\setlength\belowdisplayskip{-10pt}
\begin{equation*}
\max_{x_i\in \X_i}\int_0^t\langle v_i(x(s)),x_i\rangle ds-\int_0^t\langle v_i(x(s)),x_i(s)\rangle ds\leq h_{\max, i}-h_{\min, i}.
\end{equation*}
\label{lem:support}
\end{lemma}
\begin{proof}[Proof of Lemma~\ref{lem:support}]
We begin by deriving a bound for every fixed $x_i\in \X_i$ on the expression
\begin{equation}
\int_0^t\langle v_i(x(s)),x_i\rangle ds.
\label{eq:seq1}
\end{equation}
From the definition of the utility dynamics given by 
\begin{equation*}
y_i(t) =\int_0^t v_i(x(s)) ds,
\end{equation*}
an equivalent representation of~\eqref{eq:seq1} is
\begin{equation}
\int_0^t \langle v_i(x(s)),x_i\rangle ds =\langle y_i(t), x_i\rangle.
\label{eq:seq2}
\end{equation}
From Fenchel's inequality $\langle y_i(t), x_i\rangle\leq h_i(x_i)+h_i^\ast(y_{i}(t))$ and by definition $h_i(x_i) \leq h_{\max, i}$. Combining each inequality with~\eqref{eq:seq2}, we get
\begin{equation}
\int_0^t \langle v_i(x(s)),x_i\rangle ds \leq h_i^\ast(y_{i}(t))+h_{\max, i}.
\label{eq:seq3}
\end{equation}
We now work on obtaining a bound for $h_i^\ast(y_{i}(t))$. Observe that by definition
\begin{align*}
h^\ast_i(y_i(t))&=\langle y_i(t),Q_i(y_i(t))\rangle-h_i(Q_i(y_i(t)))\\&=\int_0^t \langle v_i(x(s)),Q_i(y_i(s))\rangle ds-h_i(Q_i(y_i(t))).
\end{align*}
Now define the function
\[\phi:(z,t)\mapsto \int_0^t \langle v_i(z(s)),z(s)\rangle ds-h(z(t)).\]
For any fixed $t\geq 0$, we can verify by the maximizing argument (see, e.g., \cite[\S2.7]{shalev2012online}) that $Q_i(y_i(t))$ maximizes $\phi(\cdot,t)$, so we can apply the envelope theorem~\citep{milgrom2002envelope}
to take the partial derivative of $\phi(Q_i(y_i(t)),t)$ with respect to the argument $t$. In doing so, we get
\[\frac{d}{dt}h^\ast_i(y_i(t))=\frac{\partial }{\partial t}\phi(Q_i(y_i(t)),t))= \langle v_i(x_i(t)),Q_i(y_i(t))\rangle.\]
Then, integrating the equation overhead, we obtain
\[h_i^\ast(y_i(t))-h_i^\ast(y_i(0))=\int_0^t\langle v_i(x(s)),Q_i(y_i(s))\rangle ds.\]
Since $h_i^\ast(y_i(0))=-h_{\min, i}$, we get the bound
\[ h_i^\ast(y_i(t))\leq \int_0^t \langle v_i(x(s)),Q_i(y_i(s))\rangle ds-h_{\min, i}.\]
Finally, combining the previous equation with~\eqref{eq:seq3}, we conclude the stated result of
\begin{equation*}
\begin{aligned}
\max_{x_i\in \X_i}&\int_0^t\langle v_i(x(s)),x_i\rangle ds-\int_0^t\langle v_i(x(s)),x_i(s)\rangle ds\leq h_{\max, i}-h_{\min, i}.
\end{aligned}
\end{equation*}
\end{proof}
We now return to proving Proposition~\ref{thm:regretbdd}.
By definition $u_i(x)=\langle v_i(x),x_i\rangle$, which means we can directly apply Lemma~\ref{lem:support} to the regret definition. We now do so and obtain the stated result:
\begin{align*}
\mathrm{Reg}_i(t)&=\max_{x_i\in \X_i}\frac{1}{t}\int_0^t (u_i(x_i, x_{-i}(s))-u_i(x(s)))ds \\
&= \max_{x_i\in \X_i}\frac{1}{t}\int_0^t\langle v_i(x(s)),x_i-x_i(s)\rangle ds\\
&\leq \frac{h_{\max, i}-h_{\min, i}}{t}=\frac{\Omega_i}{t}.
\end{align*}

\label{appsecs:lp}

\section{Supplementary Experiments and Details}
\label{appsecs:experiments}
The focus of this section is to provide supplementary simulations and details on our experimental methodology. 
We provide further simulations of the time-evolving generalized rock-paper-scissors game in Appendix~\ref{appsec:further_rps}. Then, in Appendix~\ref{appsecs:multiplayer}, we present simulations on a 5-player rescaled zero-sum polymatrix game. In Appendix~\ref{appsecs:largescalesim}, we provide simulations for larger systems.
Finally, Appendix~\ref{appsecs:reproduce} contains a description of our simulation environment and methods to allow for easy reproduction of the experimental results.

\subsection{Simulations of Time-Evolving Generalized Rock-Paper-Scissors Game}
\label{appsec:further_rps}
As mentioned in the introduction, there is a breadth of work studying the emergence of recurrent behavior of replicator dynamics in network zero-sum games~\cite{piliouras2014persistent,piliouras2014optimization,boone2019darwin, mertikopoulos2018cycles,nagarajanchaos20,perolat2020poincar}. The canonical method to prove such a result is showing that the Kullback-Leibler (KL) divergence between the replicator dynamics trajectory and the Nash equilibrium remains constant.  However, it is not clear how to apply this proof method when the game is no longer static since the Nash equilibrium of the game is not fixed. This is in fact a key technical challenge that we overcome in this work. To illustrate this, we consider the time-evolving generalized rock-paper-scissors game proposed by~\citet{mai2018cycles}. We show the evolution of the Nash equilibrium over time in Figure~\ref{fig:sub1nash}, the evolution of the population strategy vector in Figure~\ref{fig:sub2y}, and the KL divergence between the evolving equilibrium and the replicator trajectory in Figure~\ref{fig:sub3kl}. Clearly, the Nash equilibrium is no longer static and furthermore the KL divergence is not a constant of motion. This precludes the opportunity to follow standard proof techniques for showing replicator dynamics are recurrent in time-evolving games.
\begin{figure}[h]
\centering
\hfill
 \subfloat[][Evolution of Nash 
 ]{\includegraphics[width=0.32\linewidth]{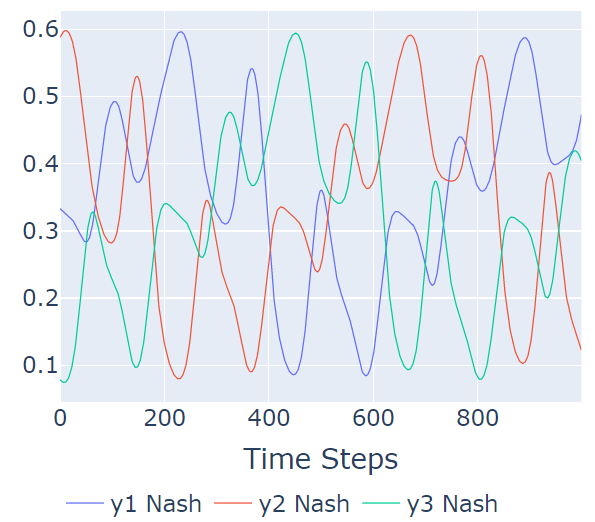}\label{fig:sub1nash}}\hfill
 \subfloat[][Evolution of $y$-player
 strategies]{\includegraphics[width=0.32\linewidth]{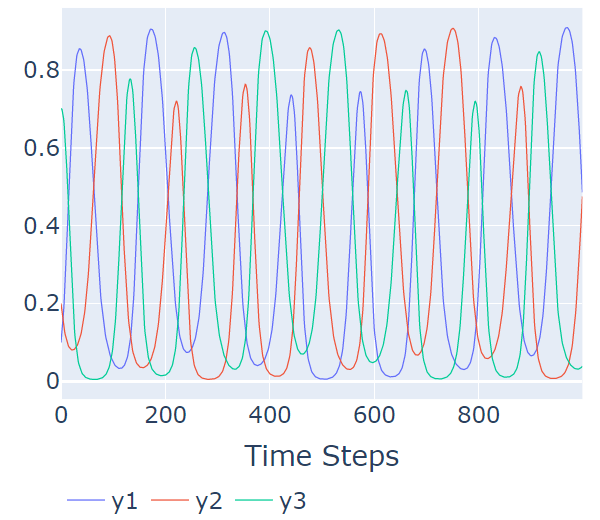}\label{fig:sub2y}}\hfill
 \subfloat[][KL-divergence]{\includegraphics[width=0.34\linewidth]{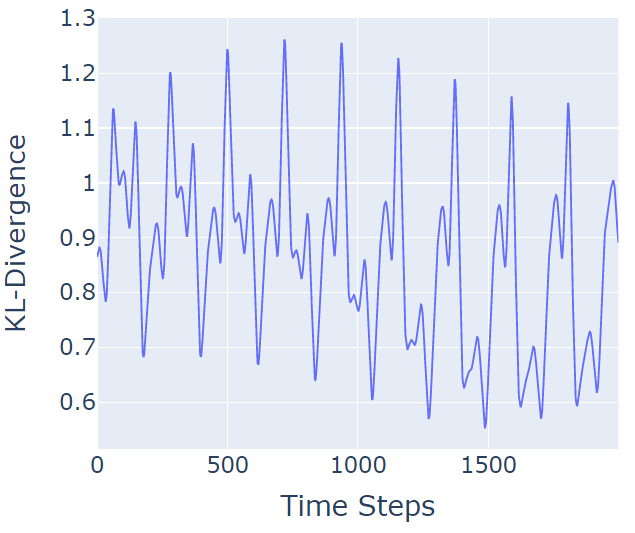}\label{fig:sub3kl}}\hfill
  \caption{Nash equilibrium of the time-evolving game, evolving strategies and the KL divergence between the Nash and strategies, from left to right, under replicator dynamics for the time-evolving generalized RPS mode \cite{mai2018cycles}.}
    \label{fig:klevolve}
\end{figure}

We solve this technical challenge by reducing time-evolving dynamics to a static polymatrix game and then proving a constant of motion. Indeed, for the time-evolving generalized rock-paper-scisscors game, we can verify that the constant of motion from Corollary~\ref{cor:newconstant} holds empirically.   Figure~\ref{fig:kl_rps} shows the weighted sum of KL-divergences from the equilibrium of the rescaled zero-sum game we obtain from our reduction to the strategy of each player along the replicator trajectory is constant. In other words, while a constant of motion exists for the time-evolving generalized rock-paper-scissors game, it not the obvious choice.
\begin{figure}[h!]
    \centering
    \includegraphics[width=0.75\linewidth]{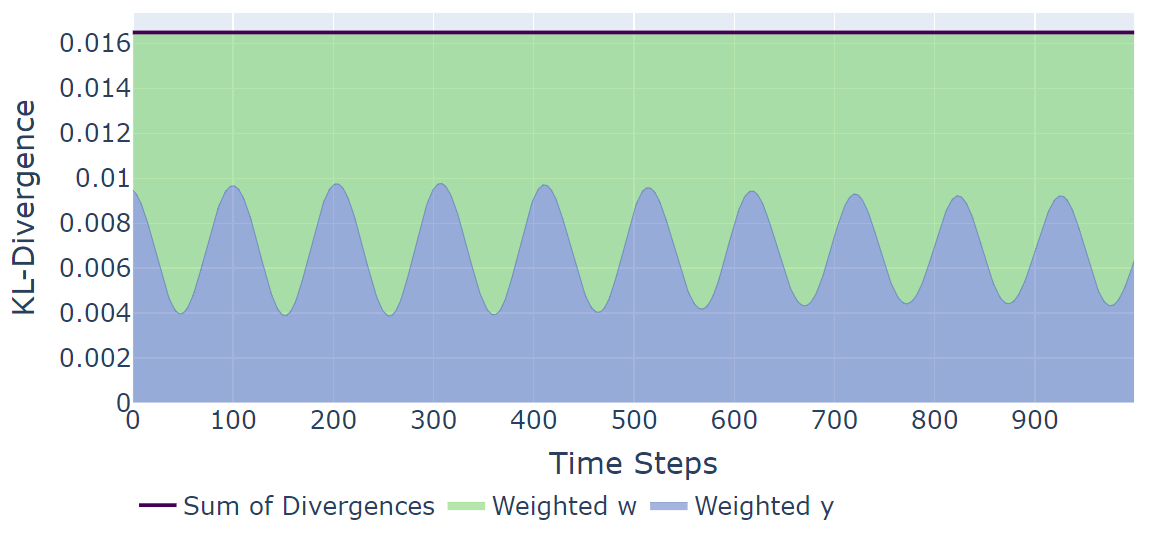}
    \caption{Constant sum of KL-divergence for time-evolving generalized rock-paper-scissors game.}
    \label{fig:kl_rps}
\end{figure}

In Figure~\ref{fig:psection2}, we present a Poincar\'e section developed from simulating 10 trajectories of initial conditions $\{[0.5, 0.01k, 0.5-0.01k, 0.5, 0.25, 0.25]\}_{k=1}^{10}$ and taking the points that intersect the hyperplane $y_2-y_1-w_2+w_1=0$. 
In Figure~\ref{fig:psectionside}, we show another example of a Poincar\'e section by simulating 10 trajectories using initial conditions $\{[1/3, 0.03k, 2/3-0.03k, 1/3, 1/3, 1/3]\}_{k=1}^{10}$ and marking where the trajectories intersected the hyperplane $y_2+y+1+w_2+w_1=4/3$. The intersection points indicate the system is quasi-periodic since they lie on closed curves.

To visualize the multidimensional system behavior while retaining the maximum amount of information, we generated Figure~\ref{fig:embedding1}, which transforms the 3-dimensional data for players $y$ and $w$ respectively to two dimensions. To be precise, the transformations are given as follows:
\begin{align*}
&y_1' = \frac{\sqrt{2}}{2} y_3 - \frac{\sqrt{2}}{2} y_2, \qquad y_2' =  - \frac{1}{\sqrt{6}} y_3 - \frac{1}{\sqrt{6}} y_2 + \frac{\sqrt{2}}{\sqrt{3}} y_1 \\
&w_1'= \frac{\sqrt{2}}{2} w_3 - \frac{\sqrt{2}}{2} w_2, \quad w_2' = - \frac{1}{\sqrt{6}} w_3 - \frac{1}{\sqrt{6}} w_2 + \frac{\sqrt{2}}{\sqrt{3}} w_1
\end{align*}
The 4-dimensional system is now visualized in the 3-dimensional plane, with color acting as the final dimension. The simulations are run for a range of initial conditions to show that when we start closer to the interior fixed point, trajectories are bounded closer to zero. These simulations were then compiled into an animation, which can be found in the supplementary code repository. Figure~\ref{fig:embedding1} and the corresponding animation are analogous to Figure~\ref{fig:sub3} and its corresponding animation, but the transformation method allows for visualization of all 6 dimensions instead of a subset of 4 dimensions.

\begin{figure}[!ht]
    \centering
    \subfloat[][Constant of motion]{\includegraphics[width=0.45\textwidth]{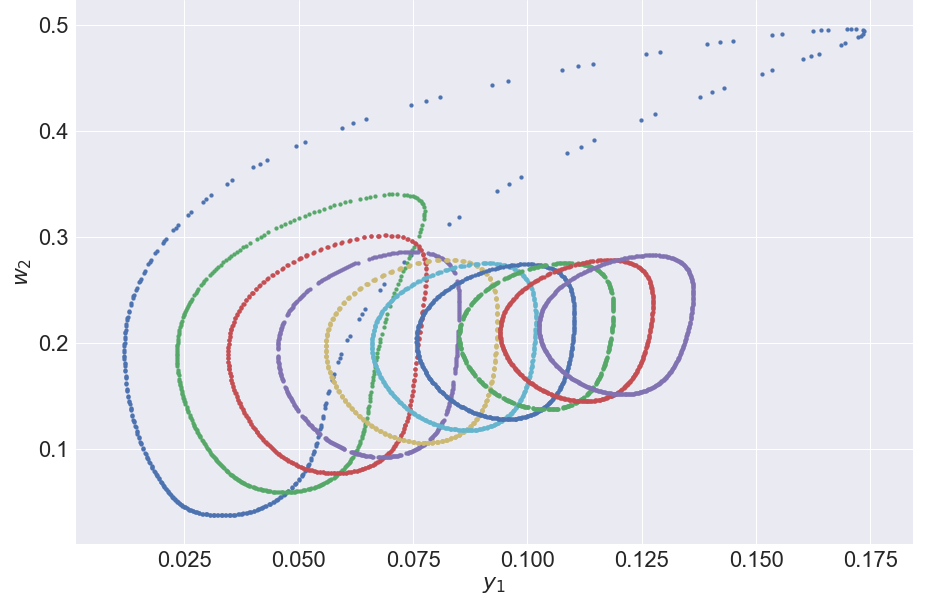}\label{fig:psection2}}\hfill
    \subfloat[][Poincar\'e section]{\includegraphics[width=0.55\textwidth]{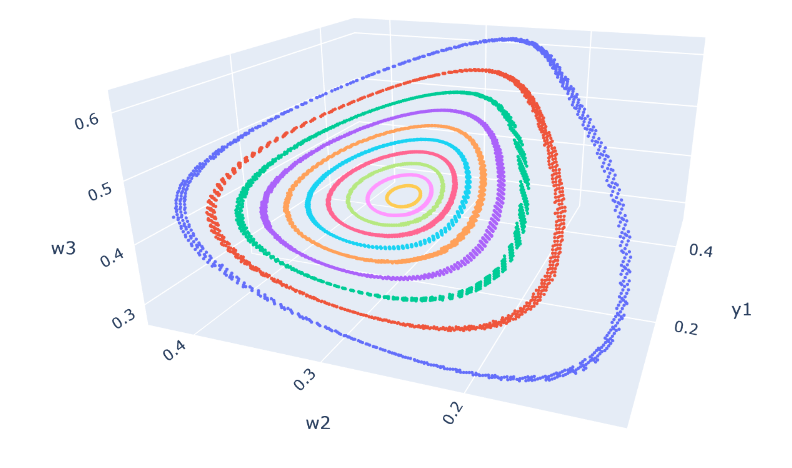}\label{fig:psectionside}}\hfill
    
    \caption{(a) 2D Poincar\'e section at $\species_2-\species_1-\weights_2+\weights_1=0$ with 10 trajectories of initial conditions $\{[0.5, 0.01k, 0.5-0.01k, 0.5, 0.25, 0.25]\}_{k=1}^{10}$, (b) Side view of Poincar\'e section at $\species_2+\species_1+\weights_2+\weights_1=4/3$ with 10 trajectories of initial conditions $\{[1/3, 0.03k, 2/3-0.03k, 1/3, 1/3, 1/3]\}_{k=1}^{10}$.}
    \label{fig:psection1}
\end{figure}

\begin{figure}[!ht]
    \centering
    \includegraphics[width=0.8\textwidth]{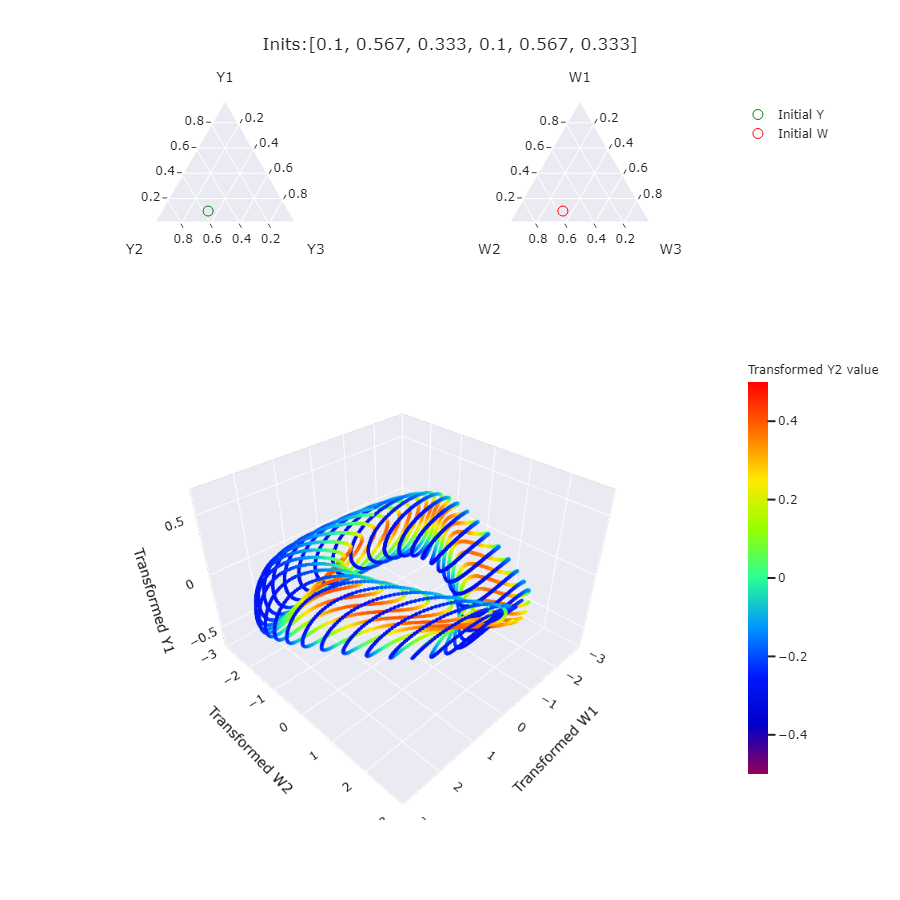}
    \caption{4D embedding of trajectories for a range of initial conditions.}
    \label{fig:embedding1}
\end{figure}

Another important property of the considered population/environment dynamics is that both the time-average vector produced by the replicator dynamics and the time-average utilities converge to the equilibrium values (see Theorem~\ref{t:avg1} of Section~\ref{sec:recurrence}). In Figures~\ref{fig:yavg} and~\ref{fig:time_avg1} we plot the time-average of the population $y$ and environment $w$ in the time-evolving generalized rock-paper-scissors game, all of which converge to $1/3$ which is the equilibrium strategy. Similarly in Figures~\ref{fig:time_avg3} and~\ref{fig:time_avg2} we plot the time-average utilities of the population
$y$ and environment $w$ converging to the equilibrium utility which in the considered instance is zero.

\begin{figure}[H]
    \centering
    \subfloat[][Time-average strategy: $\species$]{\includegraphics[width=0.24\linewidth]{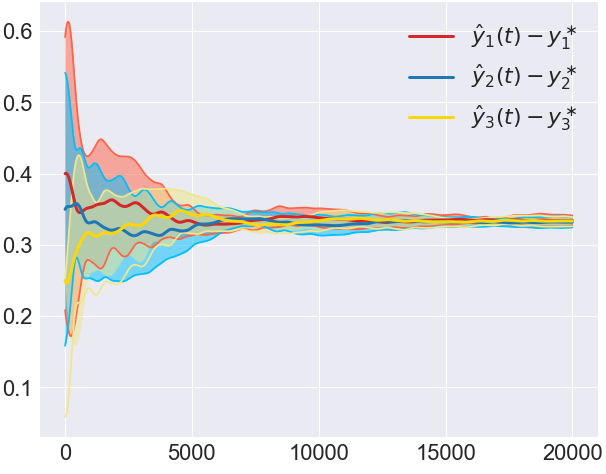}\label{fig:yavg}}\hfill
    \subfloat[][Time-average strategy: $w$]{\includegraphics[width=0.24\linewidth]{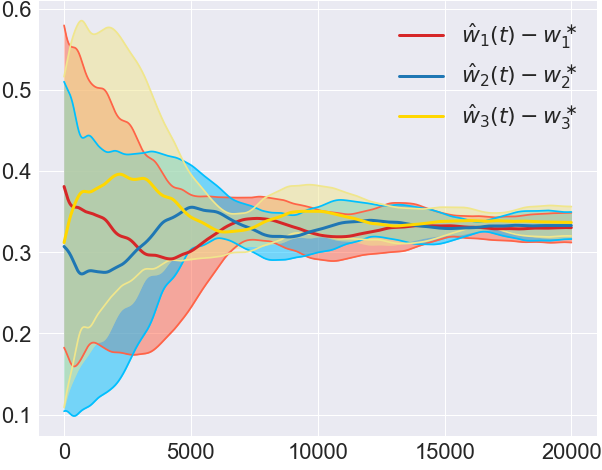}\label{fig:time_avg1}}\hfill
     \subfloat[][Time-average strategy: $\species$]{\includegraphics[width=0.255\linewidth]{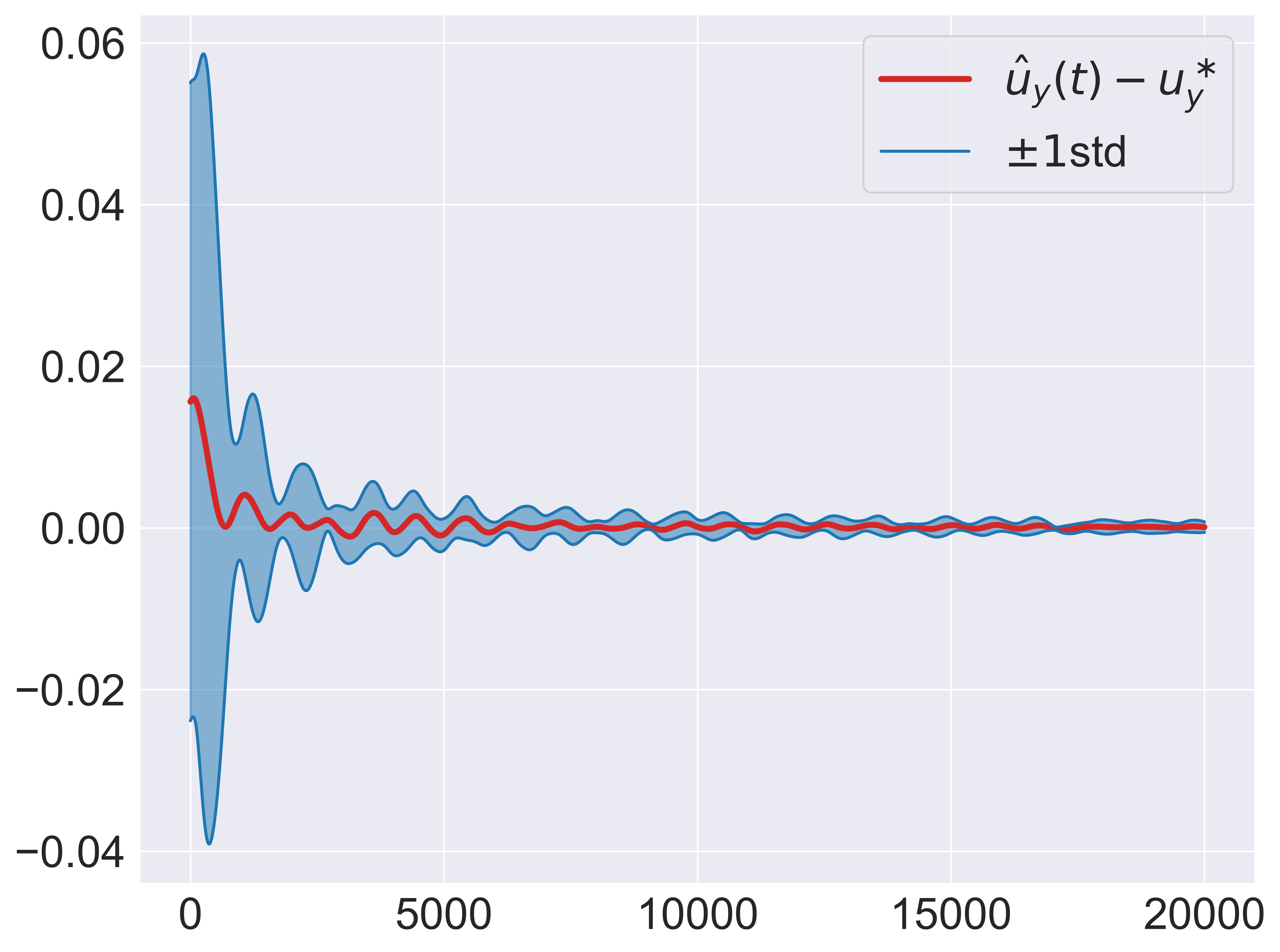}\label{fig:time_avg3}}\hfill
     \subfloat[][Time-average strategy: $w$]{\includegraphics[width=0.255\linewidth]{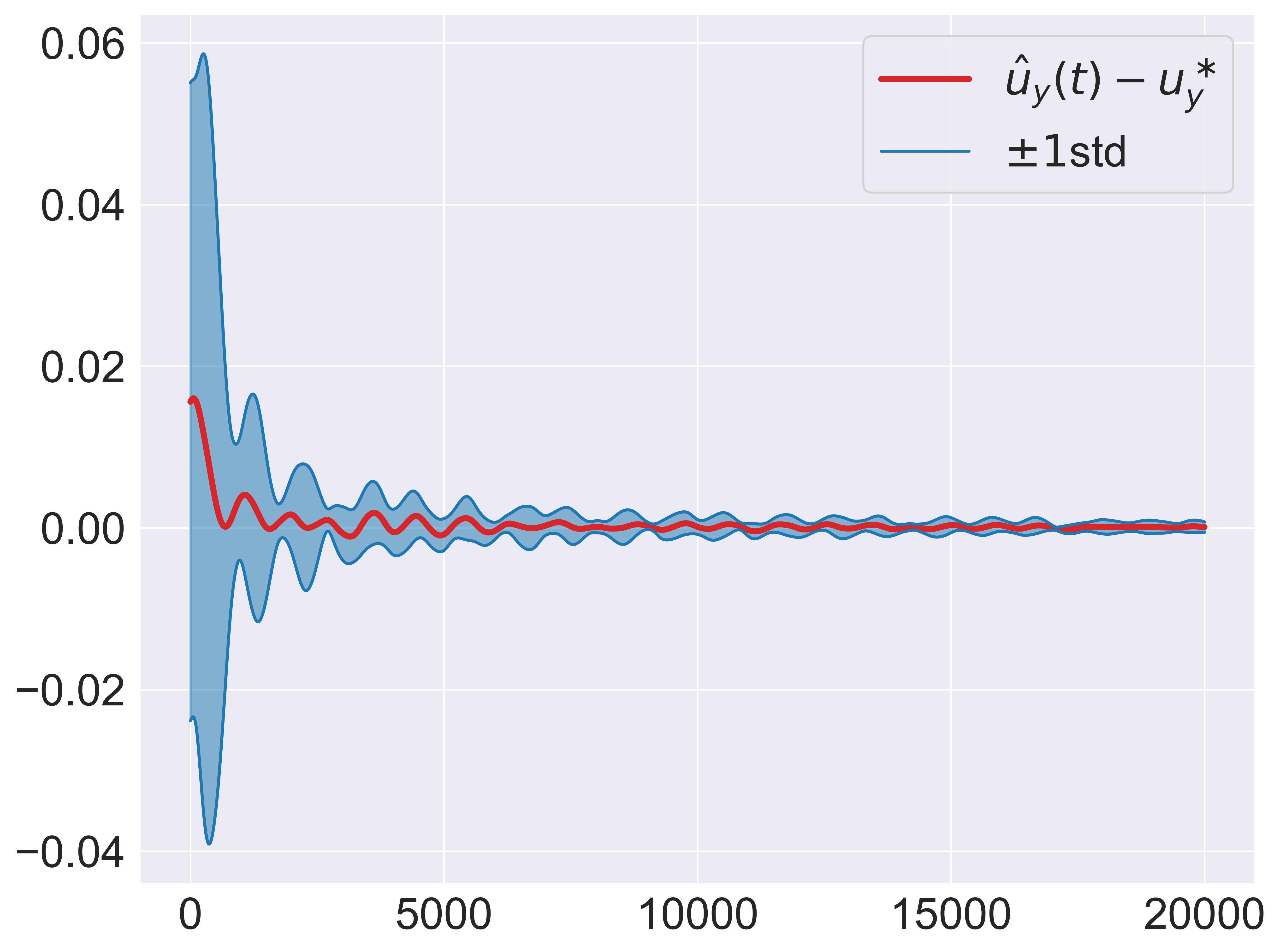}\label{fig:time_avg2}}\hfill
     \caption{(a-b) Time-average for $\species$ and $\weights$ converging to Nash, (c-d) Time-average utility converging with bounded regret.}
    \label{fig:first_timeaverage}
\end{figure}

\subsection{Simulations on 5-player Rescaled Zero-Sum Polymatrix Game}
\label{appsecs:multiplayer}
We simulate the rescaled zero-sum polymatrix game depicted in Figure~\ref{fig:graphicalgamesim1} where each player has 3 actions. 
As shown in Figure~\ref{fig:kldiv5player}, the weighted sum of Kullback-Leibler (KL) divergences of each agent's strategy from the equilibrium is a constant of motion, demonstrating the bounded orbits property of Lemma~\ref{t:invariant}.
In the simulation $\mu_1= 0.1, \mu_2=0.5, \mu_3=0.8, \mu_4=0.5$ and the initial condition was $[0.3, 0.4, 0.3, 0.2, 0.1, 0.7, 0.5, 0.3, 0.2, 0.7, 0.2, 0.1, 0.4, 0.2, 0.4]$. We also include the time averages of the trajectories and utility for player $x_3$ in Figure~\ref{fig:5player}. The plots show that the player's trajectories converge to the interior Nash equilibrium at $(1/3,1/3,1/3)$ and that the time average utility converges to the utility at this interior Nash equilibrium. 

 \begin{figure}[h]
    \centering
    \scalebox{0.95}{%
    \begin{tikzpicture}
    \begin{scope}[every node/.style={circle,thick,draw}]
        \node[fill =red!30] (x1) at (-6,2) {$x_1$};
        \node[fill =blue!30] (x2) at (-3,2) {$x_2$};
        \node[fill =red!30] (x3) at (0,2) {$x_3$};
        \node[fill =blue!30] (x4) at (3,2) {$x_4$};
        \node[fill =red!30] (x5) at (6,2) {$x_5$};
        
    \end{scope}
    \begin{scope}[>={Stealth[black]},
                 every node/.style={fill=white,circle},
                 every edge/.style={draw=black}]
        \path [<-] (x1) edge [bend right=30] node [below] {$-\I$}(x2);
        \path [<-] (x2) edge [bend right=30] node [above] {$\mu_1 \I$} (x1);
        \path (x1) edge [loop above] node {$\P$} (x1);
        
        \path [<-] (x2) edge [bend right=30] node [below] {$-\I$}(x3);
        \path [<-] (x3) edge [bend right=30] node [above] {$\mu_2 \I$} (x2);
        \path (x3) edge [loop above] node {$\P$} (x3);
        
        \path [<-] (x3) edge [bend right=30] node [below] {$-\I$}(x4);
        \path [<-] (x4) edge [bend right=30] node [above] {$\mu_3 \I$} (x3);
    
        \path [<-] (x4) edge [bend right=30] node [below] {$-\I$}(x5);
        \path [<-] (x5) edge [bend right=30] node [above] {$\mu_4 \I$} (x4);
        \path (x5) edge [loop above] node {$\P$} (x5);
    \end{scope}
    \end{tikzpicture}
    }
    \caption{Each node represents a player, with different initial strategies and values of $\mu_i$.}
    \label{fig:graphicalgamesim1}
\end{figure}
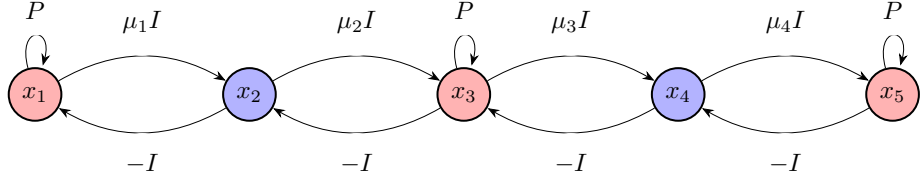
\begin{figure}[!htb]
     \centering
    {\includegraphics[width=0.95\linewidth]{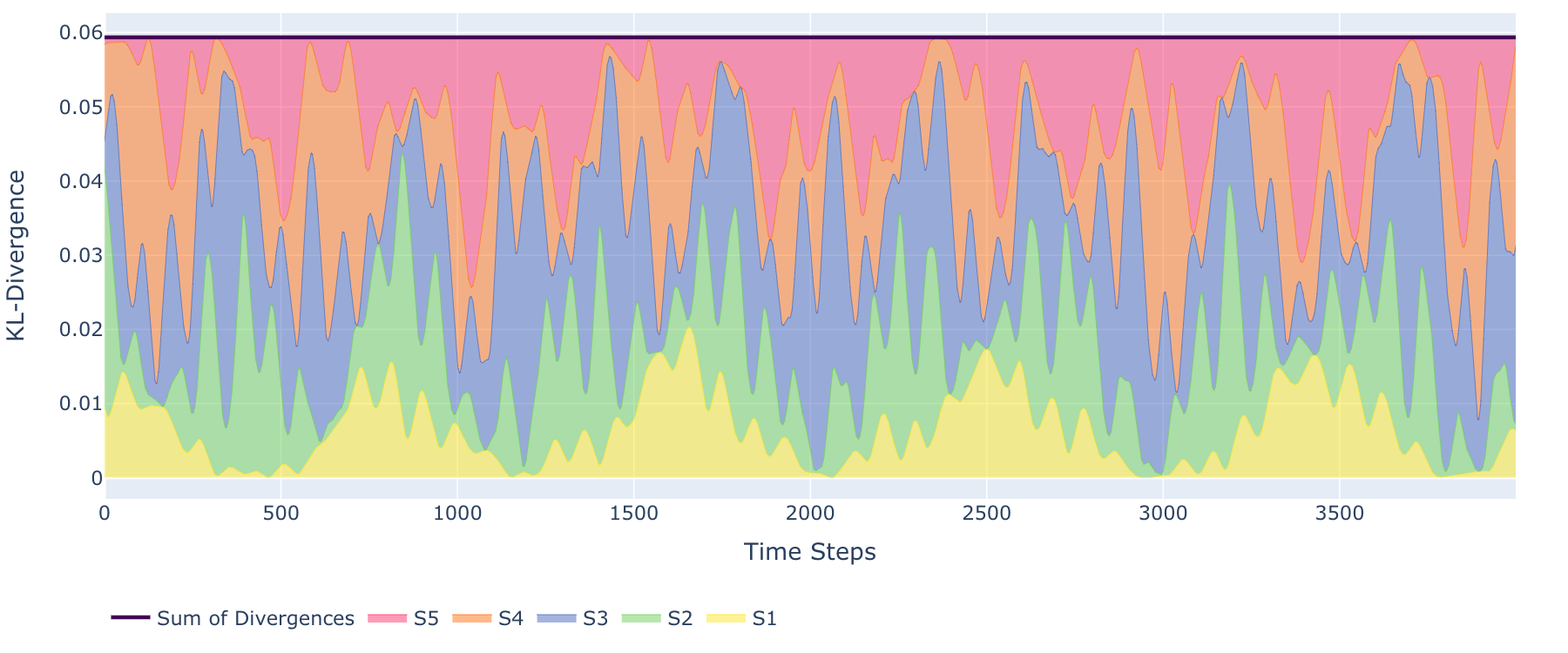}
     }
     
     \caption{Weighted KL divergence for five player time-evolving RPS game for $1000$ iterations.}
     \label{fig:kldiv5player}
 \end{figure}
 
  \begin{figure}[!htb]
     \centering
     \subfloat[][Time-average strategy: $x_3$]{\includegraphics[width=0.32\linewidth]{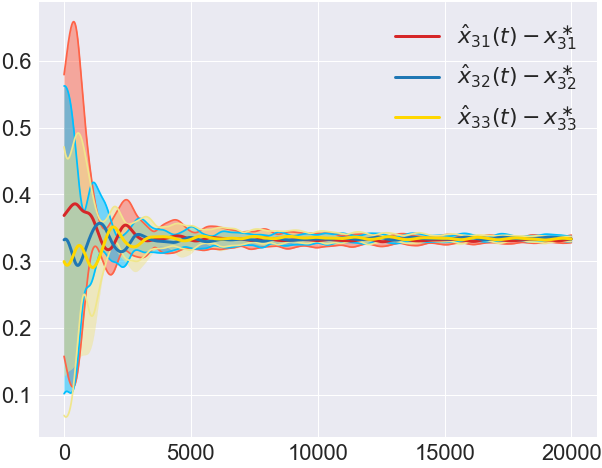}\label{fig:yavg5player}}
     \subfloat[][Time-average utility: $x_3$ ]{\includegraphics[width=0.35\linewidth]{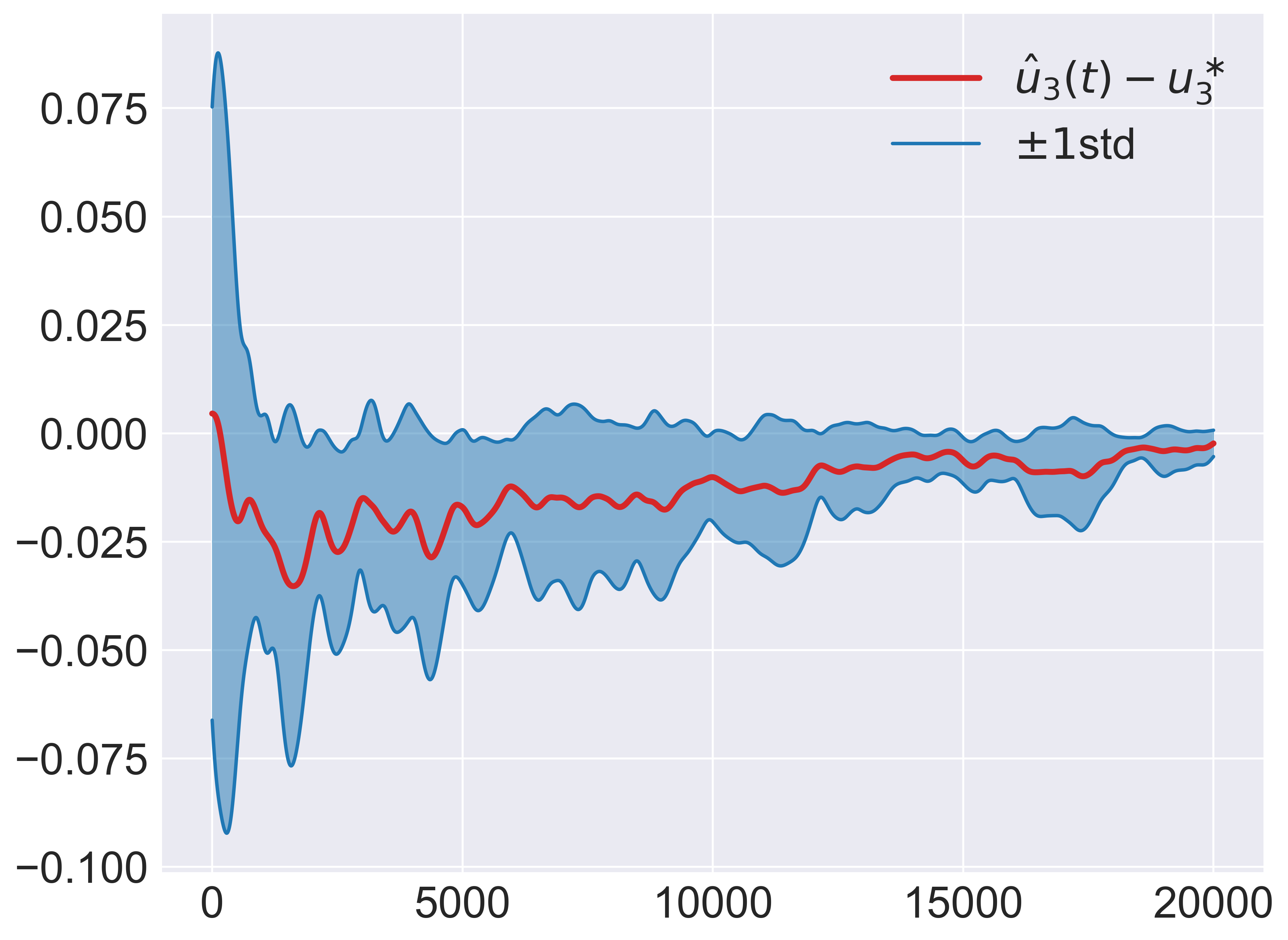}\label{fig:tavg5player}}
     \caption{ (a) Time-average trajectories for $x_3$ showing convergence to Nash. (b) Time-average utility convergence for $x_3$ player with bounded regret. Initial conditions are fixed values chosen uniformly at random on the simplex for $x_1,x_2,x_4,x_5$ and for $x_3$, they take values in $(z, 0.75-z, 0.25)$ for each $z\in\{0.1+\tfrac{2k}{30}, k\in\{0,\ldots, 9\}\}$.}
     \label{fig:5player}
 \end{figure}

\subsection{Simulations on Large-Scale Zero-Sum Polymatrix Games}
\label{appsecs:largescalesim}
To show the potential for the scalability of this theory, we simulated larger systems with more complex dynamics between players, and experimentally confirm that our theorems still hold in these contexts. 

\paragraph{Showing Poincar\'e recurrence in large-scale games.}
In order to obtain the initial conditions of this simulation, we used an 1200x1200 pixel image of Pikachu and converted that image into an 8x8 array of RGB values. Then, we convert these RGB values into a set of initial conditions for the replicator dynamics. In order to see more obvious differences between colors as the strategies evolve, we applied a sigmoid function centered at 0.5 to each RGB value.

Due to the presence of the sigmoid function, we expect to see a mostly dark or mostly bright screen whenever the strategies are far from the central value of 0.5, and after creating several animations, this hypothesis is confirmed. Indeed, in Figure \ref{fig:pikachumain} we see that the grid very quickly transforms into something that does not resemble the original at all. After a number of iterations, the recurrence property causes the Pikachu (or at least, something that looks similar to Pikachu) to reappear.

An additional point to note is that our code for simulating such large-scale rescaled zero-sum polymatrix games was refactored from the previous, smaller scaled code such that it now works for a general number of nodes N. Hence, future simulations could potentially model multiagent systems at a much wider level than shown in our work.

\paragraph{Constant KL-divergence with complex graph structures.}
In the simulations up to this point, we have been looking at rescaled zero-sum polymatrix games of a particular structure. Indeed, these simulations are extensions to the example polymatrix game as defined in Figure \ref{fig:graphicalgame}. However, our theory extends to more than just graphs of that form.  So long as the graphs are formed from the basic building blocks in Figure \ref{fig:buildingblock}, we will see similar results. We performed experiments using extensions of the `butterfly' graph shown in Figure \ref{fig:butterfly1}, where each red node is a population of species and each blue node is an environment. The connections between blue and red nodes represent bimatrix games of the form $(-I,\mu I)$ and the self-loops represent self-play zero sum games. For the simulations, we use RPS as the self-play game.

As shown in Figures \ref{fig:kldivtorus} and \ref{fig:kldivtoruslarge}, we see that despite the much more complex graph structure and many nodes, the weighted sum of divergences again sums up to a constant value. In the supplementary code, we also present an animation that shows a grid where each element represents the strategy of a node. Similarly to the Pikachu example above, we expect to see the same image after some number of iterations, but unfortunately due to the density of the graph, this would take a far larger number of iterations than our integration allows in order to achieve recurrence.

\begin{figure}[H]
     \centering
    \includegraphics[width=0.23\linewidth]{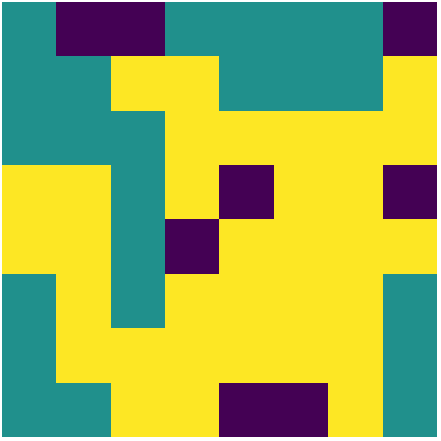}
     \caption{$8\times8$ grid of colors generated by sigmoid function}
     \label{fig:pikachu}
\end{figure}
\begin{figure}[H]
    \centering
    \includegraphics[width=0.95\linewidth]{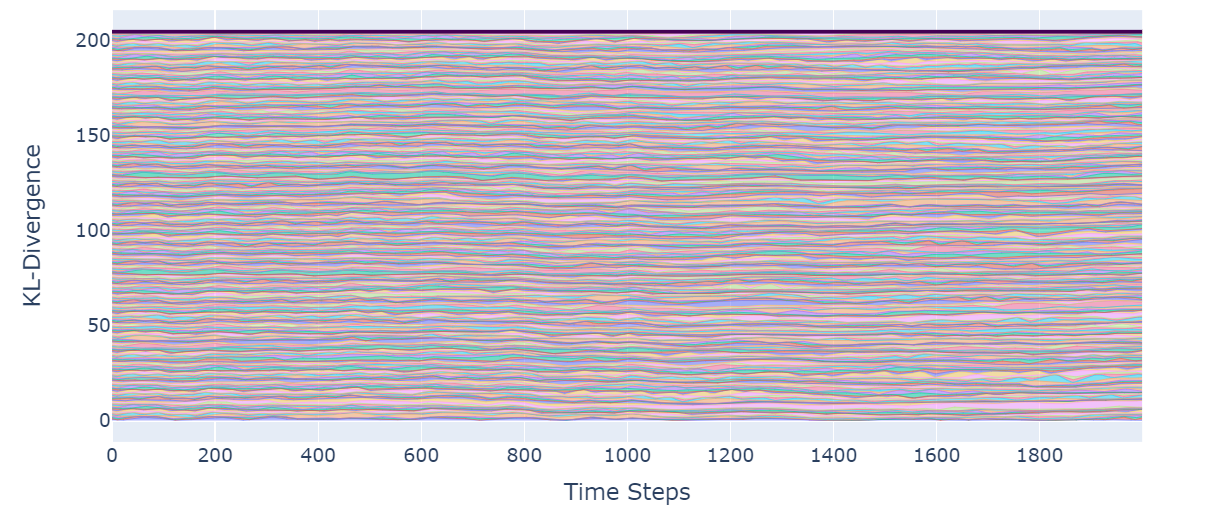}\hfill
    \caption{KL Divergences in 400-player rescaled zero-sum RPS game. Note that the sum of divergences is still constant despite the large number of nodes/players.}
    \label{fig:kldivtoruslarge}
\end{figure}

\subsection{Implementation Details}
\label{appsecs:reproduce}
The code used to generate the simulations in this paper has been compiled into a Jupyter notebook for ease of viewing. A HTML render of the notebook can be found at the following \href{https://nbviewer.jupyter.org/github/ryanndelion/egt-squared/blob/main/Evolutionary\%20Game\%20Theory\%20Squared\%20AAAI.html}{link}, while the full repository is \href{https://github.com/ryanndelion/egt-squared}{here}. 
Our code is in Python~3.6 and only requires basic scientific computing packages such as NumPy and SciPy and data visualization packages such as Matplotlib and Plotly. Most of the code in the submission has been edited so that it can easily be executed on any standard computer in a matter of minutes as it is not computationally intensive.

\end{appendices}

\end{document}